\documentclass[11pt]{article}
\usepackage[english]{babel}
\usepackage[utf8]{inputenc}
\usepackage[authoryear]{natbib}
\usepackage{geometry}
\usepackage{adjustbox}
\usepackage{setspace}
\usepackage{fontawesome5} 
\usepackage{academicons}

\usepackage{amsmath}
\usepackage{amsfonts}
\usepackage{amsthm}
\usepackage{amssymb}
\usepackage{mathrsfs}
\usepackage{bbm}
\usepackage{amstext}
\usepackage{bm}
\usepackage{mathtools}
\usepackage{todonotes}
\usepackage{enumitem}

\theoremstyle{plain}
\newtheorem{theorem}{Theorem}

\newtheorem{assumption}{Assumption}

\theoremstyle{definition}
\newtheorem{remark}{Remark}

\newtheorem{example}{Example}

\numberwithin{equation}{section}
\usepackage{chngcntr}
\counterwithin*{section}{part}
\usepackage{apptools}

\usepackage{multirow}
\usepackage{babel}
\usepackage{array}
\usepackage{rotating}
\usepackage{longtable}
\usepackage{float}
\usepackage{booktabs}
\usepackage{lscape}

\usepackage{caption}
\usepackage{subcaption}

\usepackage[ruled]{algorithm2e}

\usepackage{color}
\usepackage{graphicx}

\usepackage{xcolor}
\usepackage{listings}
\usepackage{alltt}

\definecolor{codegreen}{rgb}{0,0.6,0}
\definecolor{codegray}{rgb}{0.5,0.5,0.5}
\definecolor{codepurple}{rgb}{0.58,0,0.82}
\definecolor{backcolour}{rgb}{0.95,0.95,0.92}

\lstdefinestyle{mystyle}{
    backgroundcolor=\color{backcolour},   
    commentstyle=\color{codegreen},
    keywordstyle=\color{magenta},
    numberstyle=\tiny\color{codegray},
    stringstyle=\color{codepurple},
    basicstyle=\ttfamily\footnotesize,
    breakatwhitespace=false,         
    breaklines=true,    captionpos=b,                    
    keepspaces=true, numbers=left,                    
    numbersep=5pt, showspaces=false,                
    showstringspaces=false, showtabs=false, tabsize=2}
\usepackage{url}
\usepackage{hyperref}
\hypersetup{
    colorlinks=true,
    linkcolor=blue, 
    filecolor=[rgb]{0.6000000, 0.1058824, 0.1176471}, 
    citecolor=blue, 
    linkbordercolor=[rgb]{1.0000000, 0.7803922, 0.1725490}, 
    citebordercolor=[rgb]{0.8666667, 0.6392157, 0.0000000}, 
    filebordercolor=[rgb]{1.0000000, 0.7803922, 0.1725490}, 
    urlbordercolor=[rgb]{1.0000000, 0.7803922, 0.1725490} 
}
\definecolor{myred}{rgb}{0.6000000, 0.1058824, 0.1176471}
\definecolor{mypurple}{rgb}{0.4588235, 0.05882352, 0.4274509}

\usepackage{todonotes}

\usepackage{rotfloat}
\usepackage{pdfpages}
\usepackage{indentfirst}
\usepackage{enumitem}

\makeatletter
\def\maxwidth{ %
  \ifdim\Gin@nat@width>\linewidth
    \linewidth
  \else
    \Gin@nat@width
  \fi
}
\makeatother

\definecolor{fgcolor}{rgb}{0.345, 0.345, 0.345}

\usepackage{framed}
\makeatletter
 {\par\unskip\endMakeFramed%
 \at@end@of@kframe}
\makeatother

\definecolor{shadecolor}{rgb}{.97, .97, .97}
\definecolor{messagecolor}{rgb}{0, 0, 0}
\definecolor{warningcolor}{rgb}{1, 0, 1}
\definecolor{errorcolor}{rgb}{1, 0, 0}

\allowdisplaybreaks

\title{Distributional Effects in Censored Quantile Regressions with Endogeneity and Heteroskedasticity
\thanks{%
Xi Wang: \url{theawang@connect.ust.hk}.
We sincerely thank Songnian Chen, Ivan Fern\'{a}ndez-Val, Yingyao Hu and
participants at The 2025 World Congress of the Econometric Society for their
helpful comments. The financial support from the National
Science Foundation of China through Grant 72573071.}}
\author{
Xi Wang \\ {Institute for Economic and Social Research, Jinan University}
}
\date{\today }

\begin{document}
\maketitle

\abstract{
    
Distributional effects, captured by quantile frameworks, are well-received for
characterizing heterogeneous impacts of economic factors across the unobserved
relative ranks. Censored outcome, endogenous regressor and heteroskedastic
error are prevalent in empirical work, yet challenge the consistency of
existing quantile estimation methods.

This paper proposes a two-nested-step(TNS) estimation method for distributional effects in censored quantile models with endogeneity and heteroskedasticity. It combines the sequential analysis with the control function approach, adapting for heterogeneous distributional effects. The estimation algorithm is a two-step procedure nested with a sequence of series quantile regressions, thereby
providing applied researchers with a computationally tractable and
practically feasible tool. Monte Carlo simulation results demonstrate the good performance of our estimator in a finite sample. We apply the proposed method to estimate
heterogeneous income elasticities of households across relative ranks of commodity expenditure using data from the UK Family Expenditure Survey.\\
    \\
    Key words: Distributional effects, censored quantile regression, endogeneity, heteroskedasticity\\
    JEL code: C13,C31,C34,C36
}
\thispagestyle{empty}
\newpage \setcounter{page}{1}

\section{Introduction}

Since the seminal contribution of \citet{koenker1978regression}, quantile
regression, $Y=X^{\prime }\beta (U)$, has received substantial attention for
its flexibility in analyzing heterogeneous distributional effects of
regressors and partially accommodating multiplicative heteroskedasticity.
Unlike average effects, distributional effects $\beta (U)$ reveal how the impact of an economic factor varies
across unobserved ranks---an aspect especially appealing to researchers.
Classical quantile frameworks assume that the dependent variable is fully
observed and all regressors are independent of unobservables.

In applied work, censoring, endogeneity and heteroskedasticity are pervasive
and often arise jointly. \citet{kowalski2016censored}
analyzed heterogeneous income elasticities of demand among across high- and low-rank consumption households in a setting, where a nontrivial proportion of households make no purchases of medical care goods,
and the marginal price---the regressor of interest---is endogenous because of
simultaneity between demand and price. Conditional heteroskedasticity
is another widespread feature and a well-recognized source of heterogeneous
distributional effects. For example, high-consumption households have a wider spending dispersion and exhibit greater spending
propensity than low-consumption households, and high priced
goods display more volatility. While studies
targeting average effects have extensively discussed censoring, endogeneity,
and heteroskedasticity, these features continue to complicate the consistent
estimation of distributional effects.

In this paper, we consider a quantile regression model that accommodate endogeneity, censoring and heteroskedasticity---  complex but common empirical data features. Our paper contributes to the growing literature on censored and endogenous quantile models. Existing mean regression approaches for censored endogenous models, such as \citet{honore2004estimation}
and \citet{chen2020semiparametric}, assume constant coefficients and typically
require statistical independence between the error term and exogenous regressors. 
As a result, they cannot capture heterogeneous effects across ranks of unobserved characteristics, particularly that induced by heteroskedasticity. By contrast, standard quantile regression methods target distributional
effects to characterize the heterogeneity, but are generally inconsistent in the presence of censoring or endogenous regressors.
A substantial literature on quantile models has explored these issues
from separate perspectives. On one side, \citet{chernozhukov2005iv,chernozhukov2006instrumental,lee2007endogeneity,jun2009local,imbens2009identification,wuthrich2019closed,frolich2013unconditional}
and \citet{sasaki2025extreme} investigated endogenous quantile models.
On the other, \citet{powell1984least,powell1986censored,fitzenberger199715,
chernozhukov2002three,portnoy2003censored} and \citet{fitzenberger2007improving}
studied censored quantile models. These strands have evolved separately, and their integration is difficult in most cases.

Given the prevalence of concurrent censoring and endogeneity, developing an estimation method for distributional effects of censored quantile
models with endogenous regressors merits considerable attention. \citet{frandsen2015treatment}
analyzed a censored quantile treatment effect model with a binary
endogenous treatment. Building on this line of research, \citet{chen2018sequential,wang2021moment}
and \citet{Chen2026SSIV} proposed moment-based estimators
that require the endogenous regressor to be discrete or one-sided
bounded. Despite important advances, the bounded support restriction substantially limits their applicability to the circumstance of an endogenous regressor with an unbounded support, for instance, the logarithm of total expenditure in studies of income elasticity of demand.

To accommodate continuous endogenous regressors, \citet{blundell2007censored} extended
\citeauthor{powell1986censored}'s censored quantile regression by using the error term from
the endogenous-variable equation as a control variable to account for
endogeneity. When the endogenous regressor exhibits conditional
heteroskedasticity, the associated heavy-tailed errors generate numerous
extreme residuals, rendering their estimator highly sensitive to the
trimming threshold and creating practical challenges for applied
researchers. \citet{chernozhukov2015quantile} developed a CQIV estimator by
embedding \citeauthor{chernozhukov2002three}'s (2002) algorithm into a control function
approach. As they both emphasized, when conditional heteroskedasticity arises in the latent outcome, the maximum-likelihood estimation involved in the
selection process fails and even the algorithm may not converge. As a result, the CQIV estimator fails to consistently estimate distributional effects particularly induced by heteroskedasticity.

This paper proposes a two-nested-step(TNS) estimation method for the distributional effects in censored quantile models with an endogenous regressor and heteroskedasticity. 
The primary contribution is to identify and consistently estimate distributional effects which are heterogeneous across relative ranks of latent outcome. For these heterogeneous effects, we
solve the censoring issue by sequentially selecting the subsample, and integrate it with the control function approach. The proposed method avoids the failure of maximum likelihood estimation under heteroskedasticity, and also complements moment-based methods (e.g., \citet{chen2018sequential}) to accommodate
an unbounded endogenous regressor. Second, we suggests the conditional cumulative distribution function (CDF) serving as the control variable that accounts for endogeneity, and propose a regression based approach to estimate it. It bypasses the practical difficulties of removing excessive
outlying residuals and choosing a trimming threshold on quantile level for 
estimation of counterfactual distribution. Lastly, the estimation algorithm 
is essentially a two-step procedure nested with a sequence of series quantile 
regressions. The computation is tractable for applied researchers. For greater
generality, the TNS method can be extended to the model with non-additive
structures in both the endogenous regressor and latent outcome equations.

The remainder of the paper is organized as follows: Section 2 develops
identification and estimation of distributional effects in the censored quantile model with endogeneity and heteroskedasticity. Section 3 establishes asymptotic properties of the proposed estimator. Section 4 extends the framework to a more general model
with non-additive structures. Section 5 carries out Monte Carlo simulations. Section 6 applies the TNS method
to analyze heterogeneous expenditure elasticities of demand. Section 7
concludes this paper.

\section{Identification and Estimation}
    
For a more comprehensive assessment of the relationship between the latent outcome and economic factors, researchers are often inclined to learn distributional effects on various segments in the conditional distribution of latent outcome. The pervasive heteroskedasticity in empirical data is an importance source of heterogeneity in distributional effects. Moreover, censoring and endogeneity are common data features that require special attention in the estimation of distributional effects. To motivate the model of interest, we present two examples.

\paragraph*{Income elasticity of demand study}

The income elasticity of demand describes how a household's consumption
pattern varies with income, often proxied by total expenditure, to
reflect consumer behavior in allocating resources across commodities. It
has received substantial attention in both theoretical and empirical
research, including \citet{blundell2007semi,li2021engel}
and \citet{lopez2021estimating}. A conventional specification
takes the form $Y=R\alpha +X^{\prime }\beta +\varepsilon $ where $Y$ is the
budget share of a commodity, $R$ is the logarithm of income(total
expenditure), and $X$ is a vector of covariates. Two problems are well recognized: One is endogeneity arising from a reversal causality relationship between the expenditure on a specific commodity and income. The other is censoring on the observed budget share because a notable proportion of households do not purchase certain commodities such as alcohol or transportation services.

 In recent years, researchers tend to explore the distributional effects of income (or total expenditure) on budget shares, with particular attention to households at the high and low ends of the conditional distribution of commodity expenditure. A possible motivator for such distributional effects is widely documented heteroskedasticity in the relevant literature. For example, as \citet{lewbel2006engel} showed, the discrepancy between observed and predicted budget shares widens with total expenditure. The elasticity is thus expected to be heterogeneous across households’ relative positions in the conditional distribution of commodity expenditure.

These empirical needs make a quantile analysis, based on the specification below, a particularly appealing tool to reassess the income elasticity of demand: %
\begin{equation*}
	Y=\max \{Y^{\ast },0\}=\max \{R\alpha (\tau )+X^{\prime }\beta (\tau
	)+\varepsilon (\tau ),0\}
\end{equation*}%
for any $\tau \in (0,1)$, where $Y^{\ast }$ represents the latent budget
share of a commodity. Unlike the constant-coefficient specification, $\alpha
(\tau )$ represents the income elasticity of commodity demand at relative rank $\tau$ of commodity expenditure. Due to endogeneity, 
$Q_{\varepsilon (\tau )|X,R}(\tau )\neq 0$ where $\varepsilon (\tau
)=Y^{\ast }-R\alpha (\tau )-X^{\prime }\beta (\tau )$.

\paragraph*{Price elasticity of demand study}

The price elasticity of demand quantifies the responsiveness of the quantity
demanded to changes in the price of commodity, and is a long-standing topic of
interest in microeconomics. As price and demand exhibit simultaneity, a substantial body of research cope with the endogenous problem to estimate price elasticities for
different commodities, for instance, agricultural products (\citet{roberts2013identifying}) 
and electricity (\citet{ito2023selection}). Further, \citet{kowalski2016censored} and \citet{pryce2019alcohol}
emphasized that influences of price change ought to be heterogeneous across the
conditional distribution of quantity demanded, and the
observed demand is censored from below for certain commodities. In view of endogeneity and censoring, an appropriate quantile regression is%
\begin{equation*}
	Y=\max \{Y^{\ast },0\}=\max \{R\alpha (\tau )+X^{\prime }\beta (\tau
	)+\varepsilon (\tau ),0\},
\end{equation*}%
for any $\tau \in (0,1)$, where $Y^{\ast }$ denotes the latent consumption
volume, $Y$ is the observed consumption volume and $X$ is a vector of
covariates. The logarithm of marginal price $R$ is endogenous with $Q_{\varepsilon (\tau
	)|X,R}(\tau )\neq 0$. The distributional effects $\alpha (\tau )$ characterize heterogeneity in the response to price changes across individuals situated at different conditional deciles of the expenditure distribution.


\bigskip
\textbf{Notations}: Throughout this paper, $Q_{Y^{\ast }|X}(\tau |X)$ denotes the $\tau $-th conditional quantile of $Y^{\ast }$ given $X$. $F_{Y^{\ast
	}|X}(y|x)$ is the conditional CDF of $Y^{\ast }$ given $X=x$. Let $f_{Y^{\ast }|X,R,V}(y|x,r,v)$ and $f_{Y^{\ast}|X,R,V}^{\prime }(y|x,r,v)$ denote the conditional density of $Y^{\ast }$ and\ its first-order derivative with respect to $y$, respectively. The minimum eigenvalue and the maximum eigenvalue are denoted by mineig$(\cdot)$ and maxeig$(\cdot)$.

\subsection{Model and identification}
Consider a general censored quantile model with endogeneity, which
is the same as \citet{chernozhukov2015quantile}:%
\begin{eqnarray}
	Y &=&\max \{Y^{\ast },0\},    \label{censor1} \\
	Y^{\ast } &=&m_{Y^{\ast }}(X,R,V,U),  \label{latent1} \\
	R &=&m_{R}(Z,V),  \label{endo1}
\end{eqnarray}
\noindent where $Y$ denotes the observed outcome, $Y^{\ast }$ is the latent
outcome and $X$ is a $K_{X}\times 1$ vector of exogenous regressors. $R$ is
a continuous endogenous regressor, which can be extended to a vector of
endogenous regressors. $Z$ is a $K_{Z}\times 1$ vector of instrumental variables containing $X$ and additional variables excluded from the
latent outcome equation. The structural functions $m_{Y^{\ast }}(x,r,v,\cdot
)$ and $m_{R}(z,\cdot )$ are unknown and strictly monotone in the fourth and
the second arguments, respectively.

Two unobserved variables, $U$ and $V$, work together to capture unobserved characteristics in the latent outcome. Conditional on $(X,R,V)$, $U$ serves as determining individuals' relative ranks in term of the latent outcome, because it can be normalized to be uniformed on the unit interval, i.e., $U|X,R,V\sim U(0,1)$. Differently, $V$ captures the dependence between the endogenous regressor and the unobserved heterogeneity in the latent outcome. The control variable $V$ is responsible for endogeneity. Conditional on $Z$, $V$ can also be normalized to be standard uniform, i.e., $V|Z\sim U(0,1)$.

In model (\ref{censor1})-(\ref{endo1}), the outcome is subject to censoring; the regressor of interest is endogenous; and the general forms of $m_{Y^{\ast }}(x,r,v,\tau
)$ and $m_{R}(z,v )$ admit multiplicative heteroskedasticity. In addition to the Skorohod representation, an equivalent representation of (\ref{latent1}) is%
\begin{equation*}
	Q_{Y^{\ast }|X,R,V}(\tau |X,R,V)=m_{Y^{\ast }}(X,R,V,\tau )
\end{equation*}%
for any $\tau \in (0,1)$. This quantile function representation makes explicit how the regressors $(X,R)$ shift the entire conditional distribution of $Y^{\ast}$ across relative ranks $\tau$, thereby facilitating the analysis of distributional effects in the quantile regression literature. 

Define $%
D_{\tau }(\eta )=1\{Q_{Y^{\ast }|X,R,V}(\tau |X,R,V)>\eta \}$ for any $\eta
\geq 0$\footnote{In practice, $\eta$ is zero or a small value.} where $1\{\cdot\}$ signifies the indicator function. To illustrate the identification strategy gradually, we temporarily treat $D_{\tau }(\eta )$ as observed and cope with its unobservability later. We impose a set of assumptions below:

\setcounter{assumption}{0}
\renewcommand{\theassumption}{ID1}
\begin{assumption}[Identification]\label{assum:ident1}
$\,$ 	

(i) (Monotonicity and Continuous Distribution) \textit{$m_{R}(z,v)$ is strictly monotone in $v$. The distribution of $V$ given $Z=z$ is absolutely continuous with a strictly increasing CDF in $v$. $F_{Y^{\ast }|X,R,V}(y|x,r,v)$ and its inverse are strictly increasing.}
    
(ii) (Exclusion) \textit{$U|R,Z\sim U|X,R,V\sim U(0,1)$ where $Z=(X,Z_{1})$ with $Z_{1}$ excluded from the latent outcome equation (\ref{latent1}).}

(iii) (Differentiability and Relevance) \textit{$Q_{Y|R,Z,D_{\tau }(\eta )=1}(\tau
|r,z)$ is continuously differentiable in $(r,z)$, and $Q_{R|Z}(v|z)$ is continuously differentiable in $z$ with $\nabla_{z_{1}}Q_{R|Z}(v|z) \neq 0$ for all $(\tau,v)$.}
\end{assumption}

In quantile regression framework, the strict monotonicity conditions are standard requirements for Skorohod representation to guarantee well-defined conditional quantile functions. In addition, $V$ given $Z$ can be normalized to have a standard uniform distribution and then identified as $F_{R|Z}(R|Z)$ under in Assumption ID1(i). 
 Assumption ID1(ii) posits an exclusion restriction, similar to \citet*{lee2007endogeneity}, implying that $Z_{1}$ does not affect the unobserved characteristics directly but through $(X,R,V)$. It can be satisfied if $U$ is independent of $Z_{1}$ given $(X,R,V)$, since $U|R,Z\sim U|X,R,Z_{1},V\sim U|X,R,V $. Besides of differentiability,
Assumption ID1(iii) imposes a relevance condition in that the excluded
instrumental variables exert non-trivial influences on the endogenous
regressor. Theorem \ref{Thm1} establishes identification of the distributional effects.

\begin{theorem}\label{Thm1}
	Under the model (\ref{censor1})-(\ref{endo1}), suppose Assumption ID1 holds.
	Then for any $v,\tau \in (0,1)$,%
	\begin{eqnarray*}
		\nabla _{r}Q_{Y^{\ast }|X,R,V}(\tau |x,r,v) &=&\nabla _{r}Q_{Y|R,Z,D_{\tau }(\eta
			)=1}(\tau |r,z)+\nabla _{z}Q_{R|Z}(v|z)\left[ \nabla
			_{z}Q_{R|Z}(v|z)^{\prime }\nabla _{z}Q_{R|Z}(v|z)\right] ^{-1}\\
		&&\times \nabla _{z}Q_{Y|R,Z,D_{\tau }(\eta )=1}(\tau |r,z)^{\prime },
	\end{eqnarray*}%
	where $\nabla _{r}$ and $\nabla _{z}$ denote partial derivatives with
	respect to $r$ and $z$, respectively..
\end{theorem}

While model (\ref{censor1})-(\ref{endo1}) offers a comprehensive insight into identification of the censored quantile model with an endogenous regressor, a fully nonparametric specification is rarely practical in
empirical applications due to the curse of dimensionality in estimation. A
semiparametric specification, by contrast, is more comfortable to researchers, owning to its tractability and interpretability. In particular, a partial linear form of $Q_{Y^{\ast}|X,R,V}(\tau |X,R,V)$ makes distributional effects of interest straightforward to be characterized by quantile coefficients. This specification is commonly used in the estimation of quantile regressions, since one can directly compare estimates via the corresponding quantile coefficients. Moreover, an appropriate specification on the endogenous regressor is desirable for convenient estimation of the control variable that accounts for endogeneity, while maintaining robustness to multiplicative heteroskedasticity. In this regard, an additive structure offers a reasonable setting.

These considerations motivate us to further consider a semiparametric model similar to \citet{blundell2007censored}:
\begin{eqnarray}
	Y &=&\max \{Y^{\ast },0\},   \label{censor2} \\
	Y^{\ast } &=&X^{\prime }\beta (U)+R\alpha (U)+m(V,U),  \label{latent2} \\
	R &=&h(Z)+\sigma (Z)F_{\epsilon ^{\ast }}^{-1}(V) \label{endo2}
\end{eqnarray}%
where $m(v,\tau)$ is an unknown function and $h(z)$ may be a linear or nonparametric function. The heteroskedastic error in (\ref{endo2}) has a multiplicative form, $%
\sigma (z)\epsilon ^{\ast }$, where $\epsilon ^{\ast }=F_{\epsilon ^{\ast}}^{-1}(V)$ and $\sigma ^{2}(z)$ is the variance function with $\sigma
(z)>0$. (\ref{latent2}) imposes a partial linear form on the conditional quantile function of $Y^{\ast }$,
$Q_{Y^{\ast}|X,R,V}(\tau |X,R,V)=X^{\prime }\beta (\tau )+R\alpha (\tau )+m(V,\tau )$,
so that $\beta (\tau )$ and $\alpha (\tau )$ characterize the distributional
effects of $X$ and $R$, respectively. \citet{blundell2003endogeneity,blundell2007censored} and \citet*{lee2007endogeneity} have extensively discussed the primitive conditions to justify such a specification. Although $(X,R)$ are additively separable from $V$, the model (\ref{censor2})-(\ref{endo2}) remains sufficiently general to accommodate censoring, endogeneity, and multiplicative heteroskedasticity, and is compatible with many empirically specifications.
\begin{example}\label{ex1}
	\begin{eqnarray*}
	Y &=&\max \{X^{\prime }\beta +R\alpha +\varepsilon ,0\}, \\
	R &=&Z^{\prime }\gamma +\epsilon ,
	\end{eqnarray*}%
	where $\epsilon =\sqrt{Z^{\prime }\tilde{\gamma}}\Phi ^{-1}(V)$ with $\Phi
	^{-1}(\cdot )$ denoting the inverse function of the standard normal
	distribution and $V\sim U(0,1)$ independent of $Z$; and $\varepsilon =\rho
	\Phi ^{-1}(V)+\sqrt{1-\rho ^{2}}(X^{\prime }\tilde{\beta}+R\tilde{\alpha}%
	)\Phi ^{-1}(U)$ with $U\sim U(0,1)$ independent of $(X,R,V)$. The distributional effects are $\beta (\tau )=\beta +%
\sqrt{1-\rho ^{2}}\tilde{\beta}\Phi ^{-1}(\tau )$ and $\alpha (\tau )=\alpha
+\sqrt{1-\rho ^{2}}\tilde{\alpha}\Phi ^{-1}(\tau )$.
\end{example}
\begin{example}(Blundell and Powell, 2007)\label{ex2}
	\begin{eqnarray*}
		Y &=&max\{X^{\prime }\beta(\tau)+R\alpha(\tau) +\varepsilon(\tau),0\}, \\
		R &=&\pi(Z) +\epsilon ,
	\end{eqnarray*}%
	where $\varepsilon(\tau)$ satisfies the distributional exclusion restriction in that $\varepsilon(\tau)|R,Z \sim \varepsilon(\tau)|R,X,V \sim \varepsilon(\tau)|V$.
\end{example}
For identification, notice that
\begin{equation}
	\Pr (Y\leq X^{\prime }\beta (\tau )+R\alpha (\tau )+m(V,\tau )|X,R,V,D_{\tau
	}(\eta )=1)=\tau \text{ }a.s.  \label{imp}
\end{equation}%
To handle unobserved $D_{\tau }(\eta )$, we impose an initial condition that
there exists a quantile level $\tau _{0}$ such that, for all  $\tau \geq \tau
_{0} $, the conditional quantile of $Y$ is unaffected by censoring and hence coincides with that of $Y^{\ast }$, that is
$Q_{Y^{\ast }|X,R,V}(\tau |X,R,V)>0$ almost surely. This condition is mild and typically satisfied in empirical applications. For instance, the high-ability worker's wage is unlikely to be influenced by the minimum wage policy. See \citet{chen2018sequential} for more discussions. Subsequently, $(\beta (\tau ),\alpha
(\tau ),m(v,\tau ))$ can be identified from (\ref{imp}) for any $\tau \geq \tau_{0} $. 

When $\tau <\tau _{0}$,
let $\tau _{L}$ denote the lowest quantile level at which the available data is
adequate for regression. Partition the quantile interval $[\tau _{L},\tau _{0})$
to a decreasing sequence of quantile points, $\tau _{0}>\tau _{1}>\cdots >\tau _{L}$.
At $\tau _{1}$, we have%
\begin{equation*}
	D_{\tau _{1}}(\eta )\geq 1\{Q_{Y^{\ast }|X,R,V}(\tau _{0}|X,R,V)>\tilde{\eta}%
	\}\equiv D_{\tau _{1}}(\tau _{0},\tilde{\eta})  a.s.,
\end{equation*}%
if $\tilde{\eta}$ is chosen as specified in the proof of Theorem \ref{Thm2}. As $Q_{Y^{\ast}|X,R,V}(\tau _{0}|x,r,v)$ has already been previously identified, the selector $D_{\tau _{1}}(\tau _{0},\tilde{\eta})$ is available. Based on the fact that $D_{\tau
	_{1}}(\tau _{0},\tilde{\eta})=1$ implies $D_{\tau _{1}}(\eta )=1$, we obtain
\begin{equation}
	\Pr (Y\leq X^{\prime }\beta (\tau _{1})+R\alpha (\tau _{1})+m(V,\tau
	_{1})|X,R,V,D_{\tau _{1}}(\tau _{0},\tilde{\eta})=1)=\tau _{1}\text{ }a.s.
	\label{feaimp}
\end{equation}%
Consequently, $(\beta (\tau _{1}),\alpha (\tau
_{1}),m(v,\tau _{1}))$ are identified from the implication (\ref{feaimp}). By recursively applying the strategy, we can identify $(\beta (\tau
_{l}),\alpha (\tau _{l}),m(v,\tau _{l}))$ for $l=1,\cdots ,L$ in sequence. To formalize it, we make the following assumptions and establish the identification for model (\ref{censor2})-(\ref{endo2}) in Theorem \ref{Thm2}.

\setcounter{assumption}{0}
\renewcommand{\theassumption}{ID2}
\begin{assumption}[Identification]\label{assum:ident2}
$\,$ 

(i) (Monotonicity and Continuous Distribution) \textit{ $\epsilon ^{\ast
	} $ is independent of $Z$ and continuously distributed with $F_{\epsilon
		^{\ast }}(e)$ strictly increasing and $Q_{\epsilon ^{\ast }}(%
	\tilde{v})=0$ for some $\tilde{v}$\footnote{An alternative normalization is $E(\epsilon ^{\ast })=0$.}; $\sigma (z)>0$ for all $z$; $F_{Y^{\ast
		}|X,R,V}(y|x,r,v)$ and its inverse are strictly increasing.}

(ii) (Exclusion) \textit{$U|R,Z\sim U|X,R,V\sim U(0,1)$ where $Z=(X,Z_{1})$ with $Z_{1}$ excluded from the latent outcome equation (\ref{latent2}).}

(iii) (Full Rank) \textit{The matrix $E[D_{\tau _{1}}(\tau _{0},\tilde{\eta})f_{Y^{\ast
		}|X,R,V}(Q_{Y^{\ast }|X,R,V}(\tau |X,R,V)|X,R,V)P(\tau )P(\tau )^{\prime }] $ is positive definite uniformly in $\tau $, where 	
	\begin{equation*}
		P(\tau )=\dbinom{X}{R}-\frac{E\left[ \left. D_{\tau }(\eta )f_{Y^{\ast
				}|X,R,V}(Q_{Y^{\ast }|X,R,V}(\tau |X,R,V)|X,R,V)\dbinom{X}{R}\right\vert V%
			\right] }{E\left[ \left. D_{\tau }(\eta )f_{Y^{\ast }|X,R,V}(Q_{Y^{\ast
				}|X,R,V}(\tau |X,R,V)|X,R,V)\right\vert V\right] }.
	\end{equation*}}
(iv) (Initial Condition) \textit{There exists $\tau _{0}\in (0,1)$ such
	that $\Pr (X^{\prime }\beta (\tau )+R\alpha (\tau )+m(V,\tau )>0)=1$ for 
	$\tau _{0}\leq \tau \leq \tau _{U}$.}
	
(v) (Continuity) \textit{$(\beta (\tau ),\alpha (\tau ),m(v,\tau ))$ is
	Lipschitz continuous in $\tau $ for all $v$.}
\end{assumption}

Assumption ID2(i)-(iii) are analogous to Assumption ID1. Following the literature on multiplicative
heteroskedasticity, we assume that $\epsilon ^{\ast }$ is independent of $Z$ for simplicity, which implies $V=F_{\epsilon ^{\ast }}\left( \frac{R-h(Z)}{\sigma (Z)}\right) $. The exclusion restriction under (\ref{latent2}) can be justified by the distributional exclusion restriction in \citet{blundell2007censored} which states conditional on observable covariates, $Z$ affects the unobserved characteristics only through $V$. The full rank condition in Assumption ID2 (iii) is akin to \citet{lee2007endogeneity} except that ours is formulated in a uniform version and for censored data model. The additional Assumption ID2(iv)-(v) are
introduced to address the fixed censoring on the dependent variable. Assumption ID2(iv) is the initial condition. Assumption ID2(v) requires continuity of $(\beta (\tau ),\alpha (\tau),m(v,\tau ))$ over $\tau $, ensuring that the selector can be replaced by an available counterpart based on information from a neighboring quantile level. Intuitively, as adjacent quantile levels become close, the threshold $\tilde{\eta}$ can approximate $\eta$.

\begin{theorem}\label{Thm2} 
	Under the model (\ref{censor2})-(\ref{endo2}), let Assumption ID2 hold,
	then, $(\beta (\tau ),\alpha (\tau ))$ for $\tau \in \lbrack \tau _{L},\tau
	_{U}]$ is identified.
\end{theorem}

It is remarkable that the
moment-based method developed by \citet{wang2021moment} only applies to censored quantile models in which the endogenous regressor is
discrete or, at minimum one-sided bounded. When the
endogenous regressor is unbounded, the moment-based estimator cannot be
implemented through the trimming technique, otherwise the fundamental moment condition is invalid. Our approach therefore complements the existing moment-based method in
terms of accommodating unbounded endogenous regressors.

\subsection{TNS Estimation}
For practical purpose, this subsection develops the TNS estimator for model (\ref{censor2})-(\ref{endo2}) and defers the
estimation of model (\ref{censor1})-(\ref{endo1}) to Section 4. The estimation proceeds in two steps, each comprising a sequence of series quantile regressions.

\medskip
\noindent\textit{Step 1: control variable estimation}
\medskip

To estimate the control variable, we adopt a simple regression-based approach. In principle, the TNS estimation accommodates a broad class of conditional CDF estimators—an active area of ongoing research estimator (see, e.g., \citet{cattaneo2024boundary}). 

Based on the identification result, $V_{i}$ evaluates the control function
\begin{eqnarray*}
v(r,z)=F_{\epsilon ^{\ast }}\left( \frac{r-h(z)}{\sigma (z)}\right) 
\end{eqnarray*}
at $v(r,z)=(R_{i},Z_{i})$. If the distribution of $\epsilon ^{\ast }$ is specified, $%
V_{i}=v(R_{i},Z_{i})$ can be estimated by
\begin{equation}
	\hat{V}_{i}=\hat{v}(R_{i},Z_{i})=F_{\epsilon ^{\ast }}\left( \frac{R_{i}-%
		\hat{h}(Z_{i})}{\hat{\sigma}(Z_{i})}\right) .  \label{ctrl}
\end{equation}%
Here, $\hat{h}(z)=P_{J_{Z}}(z)^{\prime }\hat{\gamma}_{\tilde{v}}$ where $%
P_{J_{Z}}(z)$ is a $J_{z}\times 1$ vector of basis functions and $\hat{%
	\gamma}_{\tilde{v}}$ is obtained from a $\tilde{v}$-th series quantile
regression, i.e.,
\begin{eqnarray*}
	 \hat{\gamma}_{\tilde{v}}=\underset{c}{\arg \min }\frac{1}{n}%
\sum_{i=1}^{n}\rho _{\tilde{v}}(R_{i}-P_{J_{Z}}(Z_{i})^{\prime }c)
\end{eqnarray*} 
with the check function $\rho _{\tilde{v}}(u)=(\tilde{v}-1\{u<0\})u$. The
estimator of $\sigma ^{2}(z)$ is given by 
\begin{eqnarray*}
	\hat{\sigma}^{2}(z)=\left( \frac{%
	\hat{Q}_{R|Z}(\tilde{v}_{2}|z)-\hat{Q}_{R|Z}(\tilde{v}_{1}|z)}{F_{\epsilon
		^{\ast }}^{-1}\left( \tilde{v}_{2}\right) -F_{\epsilon ^{\ast }}^{-1}\left( 
	\tilde{v}_{1}\right) }\right) ^{2},
\end{eqnarray*}
where $\hat{Q}_{R|Z}(\tilde{v}_{1}|z)$
and $\hat{Q}_{R|Z}(\tilde{v}_{2}|z)$ are nonparametric estimates of $%
\tilde{v}_{1}$-th and $\tilde{v}_{2}$-th quantiles of $R$ given $Z=z$. It is worth pointing out that our estimation is easily extended to the case that the distribution of $\epsilon ^{\ast }$ is left unspecified and estimated nonparametrically based on the condition that  
\begin{equation*}
	E(1\{R<r\}|Z=z)=F_{\epsilon ^{\ast }}\left( \frac{r-h(z)}{\sigma (z)}\right)
	.
\end{equation*}

At this point, one may wonder why we do not employ the quantile-regression-based
counterfactual distribution function (QR-CDF) estimator, $\hat{V}_{i}=\omega
+\int_{\omega }^{1-\omega }1\{\hat{Q}_{R|Z}(v|Z_{i})<R_{i}\}dv$. To answer
this query, notice that the QR-CDF estimator requires a cut-off $\omega $ to
avoid evaluation at extreme quantiles. Under the current setting, $\omega $
should be $o(n^{-1/2})$ to ensure a negligible asymptotic bias. However, the
series quantile regression estimator $\hat{Q}_{R|Z}(v|Z_{i})$ typically converges slowly at extreme quantiles near $\omega $ or 
$1-\omega $. Consequently, the induced estimator $\hat{V}_{i}$ converges at a rate slower than $n^{-1/4}$, thereby contradicting Assumption A4.\medskip

\medskip
\noindent\textit{Step 2: distributional effects estimation}
\medskip

Once $\hat{V}_{i}$ is available, we proceed to the main target—
estimating the distributional effects $\alpha (\tau )$ and $%
\beta (\tau )$. Before a formal analysis, a noteworthy point is several
classical censored quantile methods cannot be directly extended to a partial
linear specification of the conditional quantile function of the latent outcome. \citet{powell1986censored} exploited the invariant property of the
quantile function with respect to a monotone transformation to handle
censoring, that is, $Q_{Y|X}(\tau |X)=\max \{Q_{Y^{\ast }|X}(\tau
|X),0\}=\max \{X^{\prime }\beta (\tau ),0\}$. However, the resulting optimization problem is non-convex, which severely limits its applicability in high-dimensional settings, particularly induced by the series expansion of $m(v,\tau )$. \citet{chen2003estimation}
briefly offered a series estimator based on the moment condition: 
\begin{equation*}
	E\left[1\{X_{i}^{\prime }b+R_{i}a+P_{J}(V_{i})^{\prime }d>0\}(\tau -1\{Y_{i}\leq X_{i}^{\prime }b+R_{i}a+P_{J}(V_{i})^{\prime }d\})\dbinom{X_{i}}{R_{i}}\right]=0.
\end{equation*}%
But the estimator derived from the sample analog leads to inconsistent estimates, as $1\{X_{i}^{\prime }b+R_{i}a+P_{J}(V_{i})^{\prime }d>0\}$ tends to select the subsample incorrectly.

To this end, we propose to combine the series estimation with the sequential analysis implied by the identification strategy. Let $P_{J}(v)=(p_{1}(v),\cdots
,p_{J}(v))^{\prime }$ denote a $J\times 1$ vector of basis functions with
coefficient vector $\delta (\tau )$, so that $P_{J}(v)^{\prime }\delta (\tau
)$ approximates $m(v,\tau )$ well in the sense that the approximation error $%
r(v,\tau )=m(v,\tau )-P_{J}(v)^{\prime }\delta (\tau )$ vanishes as $%
J\rightarrow \infty $. For $\tau _{0}\leq \tau \leq \tau _{U}$, $(\beta
(\tau ),\alpha (\tau ),\delta (\tau ))$ is estimated by%
\begin{equation*}
	(\hat{\beta}(\tau ),\hat{\alpha}(\tau ),\hat{\delta}(\tau ))=\underset{%
		(b,a,d)\in \mathcal{B\times A\times D}}{\arg \min }\frac{1}{n}%
	\sum_{i=1}^{n}T_{i}\rho _{\tau }(Y_{i}-X_{i}^{\prime }b-R_{i}a-P_{J}(\hat{V}%
	_{i})^{\prime }d),
\end{equation*}%
where $\mathcal{B\times A\times D}$ is the parameter space. The trimming
variable $T_{i}$ $=1\{X_{i}\in \mathcal{\bar{X}},R_{i}\in \mathcal{\bar{R}}%
,Z_{i}\in \mathcal{\bar{Z}}\}$ is used to eliminate observations out of the bounded
set $\mathcal{\bar{X}\times \bar{R}\times \bar{Z}}$, where the bounds are
permitted to diverge as the sample size increases. It is noteworthy that the trimming technique is implemented only in the
second-step estimation and only to observed regressors, making it considerably simpler than that used in conventional control function methods which additionally trim on residuals.\footnote{Trimming on residuals requires the trimming threshold to go infinity at a slower rate than the convergence rate of residuals, which introduces a delicate tuning problem. By contrast, our estimation has no such a problem because we need not trim on $\hat{V}_{i}$.}

When censoring comes into play, that is, for $\tau _{L}\leq \tau <\tau _{0}$%
, we partition the quantile interval $[\tau _{L},\tau _{0})$ into a sequence of grid points collected by the set $\mathcal{T=}\{\tau _{0},\tau _{1},\cdots ,\tau _{L}\}$ with $\tau _{0}>\tau _{1}>...>\tau _{L}$. In empirical applications, $\tau _{L}$ is set a priori or numerically. For the quantile level $\tau _{1}$, define the feasible subsample selector as 
\begin{equation*}
	\hat{D%
}_{\tau _{1}i}=1\{X_{i}^{\prime }\hat{\beta}(\tau _{0})+R_{i}\hat{\alpha}%
(\tau _{0})+P_{J}(\hat{V}_{i})^{\prime }\hat{\delta}(\tau _{0})>\tilde{\eta}%
_{n}\}
\end{equation*}
where $\tilde{\eta}_{n}\geq \eta+\max (c/L,\sup_{\tau \in \lbrack \tau
	_{L},\tau _{U}]}\sup_{v\in \mathcal{\bar{V}}}\left\vert r_{J}(v,\tau
)\right\vert,c\max_{i}\left\vert \hat{V}_{i}-V_{i}\right\vert,c\nu _{n})$ for a positive constant $c$ and a positive
sequence $\nu _{n}\rightarrow 0$ at a rate slower than the uniform
convergence rate of $(\hat{\beta}(\tau ),\hat{\alpha}(\tau ),\hat{m}(v,\tau
))$. As $L$ grows with sample size, $\tau _{1}$ gets arbitrarily close to $\tau _{0}$. This enables us to choose $\tilde{\eta}_{n}$ slightly larger than $\eta$ so that, for observations with $\hat{D}_{\tau _{1}i}=1$, the $\tau_{1}$th quantile of $Y^{\ast }$ conditional on $(X,R,V)$ is positive. Accordingly, the TNS estimator for $(\beta (\tau _{1}),\alpha (\tau
_{1}),\delta (\tau _{1}))$ is defined as 
\begin{equation*}
	(\hat{\beta}(\tau _{1}),\hat{\alpha}(\tau _{1}),\hat{\delta}(\tau _{1}))=%
	\underset{(b,a,d)\in \mathcal{B\times A\times D}}{\arg \min }\frac{1}{n}%
	\sum_{i=1}^{n}\hat{D}_{\tau _{1}i}T_{i}\rho _{\tau _{1}}(Y_{i}-X_{i}^{\prime
	}b-R_{i}a-P_{J}(\hat{V}_{i})^{\prime }d).
\end{equation*}%

We then proceed sequentially to $\tau _{2},\cdots,\tau _{L}$. For $\tau _{l}\in \mathcal{T}$ with $l=2,\cdots ,L$, the TNS estimator of $(\beta (\tau _{l}),\alpha
(\tau _{l}),\delta (\tau _{l}))$ is%
\begin{equation*}
	(\hat{\beta}(\tau _{l}),\hat{\alpha}(\tau _{l}),\hat{\delta}(\tau _{l}))=%
	\underset{(b,a,d)\in \mathcal{B\times A\times D}}{\arg \min }\frac{1}{n}%
	\sum_{i=1}^{n}\hat{D}_{\tau _{l}i}T_{i}\rho _{\tau _{l}}(Y_{i}-X_{i}^{\prime
	}b-R_{i}a-P_{J}(\hat{V}_{i})^{\prime }d).
\end{equation*}
where $\hat{D}_{\tau _{l}i}=1\{X_{i}^{\prime }\hat{\beta}(\tau _{l-1})+R_{i}\hat{%
	\alpha}(\tau _{l-1})+P_{J}(\hat{V}_{i})^{\prime }\hat{\delta}(\tau _{l-1})>%
\tilde{\eta}_{n}\}$ which is determined by estimates at the previous quantile level. Although both $\hat{D}_{\tau _{l}i}$ and $T_{i}$ are indicator variables, they serve different purposes. $\hat{D}_{\tau _{l}i}$ selects the subsample with positive $\tau_{l}$th conditional quantile of latent outcome, whereas $T_{i}$ is used to remove outliers.
 
In general, for any $\tau \in \lbrack \tau
_{L},\tau _{0})$, there exists $l$ such that $\tau \in \lbrack \tau
_{l+1},\tau _{l})$. The feasible subsample selector is given by $\hat{D}_{\tau i}=1\{X_{i}^{\prime }\hat{\beta}(\tau _{l})+R_{i}\hat{%
	\alpha}(\tau _{l})+P_{J}(\hat{V}_{i})^{\prime }\hat{\delta}(\tau _{l})>%
\tilde{\eta}_{n}\}$
where $(\hat{\beta}(\tau _{l})^{\prime },\hat{\alpha}(\tau _{l}),\hat{\delta}(\tau _{l})^{\prime })^{\prime }$ was obtained at the preceding quantile level $\tau _{l}$. The TNS
estimator of $(\beta (\tau ),\alpha (\tau ),\delta (\tau ))$ is then
\begin{equation}\label{est} 
	\begin{aligned}
	(\hat{\beta}(\tau ),\hat{\alpha}(\tau ),\hat{\delta}(\tau )) &=\underset{%
		(b,a,d)\in \mathcal{B\times A\times D}}{\arg \min }S_{n}(b,a,d,\hat{\theta}%
	(\tau _{l}),\tau ,\tilde{\eta}_{n})  \\
	&=\underset{(b,a,d)\in \mathcal{B\times A\times D}}{\arg \min }\frac{1}{n}%
	\sum_{i=1}^{n}\hat{D}_{\tau i}T_{i}\rho _{\tau }(Y_{i}-X_{i}^{\prime
	}b-R_{i}a-P_{J}(\hat{V}_{i})^{\prime }d).  
	\end{aligned}
\end{equation}
Subsequently, the consistent estimator of $m(v,\tau )$ is
$\hat{m}(v,\tau )=P_{J}(v)^{\prime }\hat{\delta}(\tau )$ for $v\in \mathcal{\bar{V}}$,
where $\mathcal{\bar{V}}$ is the range of $V_{i}$ when 
$(X_{i},R_{i},Z_{i})\in \mathcal{\bar{X}\times \bar{R}\times \bar{Z}}$.

\begin{remark}
	Unlike moment-based methods, trimming on $R_{i}$
	in the second-step estimation does not pose a problem, because the relative ranks of unobserved heterogeneity are not shifted and thus $R_{i}$ can be treated as exogenous once $V_{i}$ is controlled for.
\end{remark}

Traditional control function approaches, such as \citet{blundell2007censored} and \citet{lee2007endogeneity},
use the error term from the endogenous regressor equation to control for
endogeneity, and require removal of outlying residuals to ensure boundedness
for the nonparametric estimation. Under heteroskedasticity, the prevalence of extreme residuals may render the estimator numerically sensitive to the choice of trimming threshold. In our estimation, trimming on $\hat{V}_{i}$ is unnecessary since $V_{i}$ is inherently bounded.

The CQIV method also copes with endogenous censored quantile regression, but 
the estimator is solved from $\min_{(b,a,d)}\frac{1}{n}\sum_{i=1}^{n}T_{i}1%
\{(X_{i}^{\prime },R_{i},P_{J}(\hat{V}_{i})^{\prime })^{\prime }\hat{\zeta}%
\geq \lambda \}\rho _{\tau }(Y_{i}-X_{i}^{\prime }b-R_{i}a-P_{J}(\hat{V}%
_{i})^{\prime }d)$, where $\hat{\zeta}$ is the maximum likelihood estimate
obtained from a binary regression in which the response variable indicates
whether $Y_{i}$ is subject to censoring. As noted by \citet{chernozhukov2002three},
$\hat{\zeta}$ is inconsistent under heteroskedasticity, thereby
preventing the CQIV method from consistently estimating heterogeneous
distributional effects. By contrast, the TNS method aims at consistently estimating distributional effects and provides a robust
estimation procedure particularly in the presence of
conditional heteroskedasticity.

To compute the TNS estimator, the following outlines the estimation algorithm. For notational convenience, let $\hat{W}_{i}=(X^{\prime },R,P_{J}(\hat{V}_{i})^{\prime })^{\prime }$ and $\hat{\theta}(\tau)=(\hat{\beta} (\tau )^{\prime },\hat{\alpha} (\tau),\hat{\delta}(\tau)^{\prime })^{\prime }$.\bigskip

\noindent\textbf{TNS Estimation Algorithm} \medskip

\noindent \textit{Step 1. Control variable estimation}\medskip

\noindent%
\hangafter 1%
\hangindent=4em%
Step 1.1: Run a $\tilde{v}$th series quantile regression
	of $R_{i}$ \ on $P_{J_{Z}}(Z_{i})$\ to obtain the
	coefficient $\hat{\gamma}_{\tilde{v}}$.\medskip
	
\noindent%
\hangafter 1%
\hangindent=4em%
Step 1.2: Run $\tilde{v}_{1}$th and $\tilde{v}_{2}$th series quantile regressions of $R_{i}$ on $P_{J_{Z}}(Z_{i})$ to obtain $\hat{\gamma}_{\tilde{v}_{1}}$ and $\hat{\gamma}_{\tilde{v}_{2}}$.\medskip

\noindent%
\hangafter 1%
\hangindent=4.5em%
Step 1.3: Construct the estimate of $\sigma ^{2}(z)$ as
$\hat{\sigma}^{2}(z)=\left( \frac{\hat{Q}_{R|Z}(\tilde{v}_{2}|z)-\hat{Q}%
	_{R|Z}(\tilde{v}_{1}|z)}{F_{\epsilon ^{\ast }}^{-1}(\tilde{v}%
	_{2})-F_{\epsilon ^{\ast }}^{-1}(\tilde{v}_{1})}\right) ^{2}$ with 
$\hat{Q}_{R|Z}(\tilde{v}_{1}|z)=P_{J_{Z}}(z)^{\prime }\hat{\gamma}_{\tilde{v}%
	_{1}}$ and $\hat{Q}_{R|Z}(\tilde{v}_{2}|z)=P_{J_{Z}}(z)^{\prime }%
\hat{\gamma}_{\tilde{v}_{2}}$.\medskip

\noindent%
\hangafter 1%
\hangindent=4.5em%
Step 1.4: Estimate $V_{i}$ by $\hat{V}%
_{i}=F_{\epsilon ^{\ast }}\left( \frac{R_{i}-P_{J_{Z}}(Z_{i})^{\prime }\hat{%
		\gamma}_{\tilde{v}}}{\hat{\sigma}(Z_{i})}\right) $.\medskip

\noindent \textit{Step 2. Distributional effects estimation}\medskip

\noindent%
\hangafter 1%
\hangindent=4em%
Step 2.1: Drop observations with extreme values of $(X_{i},R_{i},Z_{i})$.\medskip

\noindent%
\hangafter 1%
\hangindent=4em%
Step 2.2: Choose a set  $\mathcal{T=}\{\tau _{0},\tau _{1},\cdots ,\tau _{L}\}$ with $\tau _{0}>\tau _{1}>...>\tau _{L}$. Run a $\tau _{0}$th series quantile regression of $Y_{i}$ on $\hat{W}_{i}$ and obtain $\hat{\theta}(\tau _{0})$.\medskip

\noindent%
\hangafter 1%
\hangindent=4.5em%
Step 2.3: Using the subsample that $\hat{W}_{i}^{\prime }\hat{%
\theta}(\tau _{0})$ are greater than the $q_{0}$th
quantile of positive $\hat{W}_{i}^{\prime }\hat{\theta}(\tau _{0})$, run a $\tau _{1}$th series quantile regression to
obtain $\tilde{\theta}(\tau _{1})$.\medskip

\noindent%
\hangafter 1%
\hangindent=4.5em%
Step 2.4: Update to the subsample that $\hat{W}_{i}^{\prime }\tilde{%
	\theta}(\tau _{1})$ are greater than the $q_{1}$th quantile
	of positive $\hat{W}_{i}^{\prime }\tilde{\theta}(\tau _{1})$. Then, run a
	$\tau _{1}$th series quantile regression to obtain $\hat{%
	\theta}(\tau _{1})$.\medskip

\noindent%
\hangafter 1%
\hangindent=4.5em%
Step 2.5: Repeat Step 2.3-Step 2.4 for the quantile level $\tau
_{2},\cdots ,\tau _{L}$ and obtain $\hat{\theta}(\tau _{2}),\cdots
,\hat{\theta}(\tau _{L})$.\medskip

We offer several practical guidelines for implementing the TNS algorithm. Following \citet{chernozhukov2024conditional}, Step 1.2 can set $(\tilde{v}_{1},\tilde{v}_{2})=(0.25,0.75)$. Trimming in Step 2.1 may be omitted in applications. For initialization in Step 2.2, follow \citet{chen2023two} to set the initial quantile level to $\tau_{0}=0.99$ or $0.95$. In steps 2.3-2.4, once the grid step size is specified, the number of quantile grid points $L$ is implied; our numerical analysis indicates that a step size of $\tau _{l-1}-\tau _{l}=0.01$ performs well. The selector threshold $\tilde{\eta}_{n}$ may be either a constant or a statistic converging to $\eta $. In line with \citet{chernozhukov2015quantile}, we recommend the updating procedure with $(q_{0},q_{1})=(0.1,0.03)$ or $(0.05,0.01)$. Finally, if $P_{J}(v)$ is composed of raw polynomials, use the cross-validation method to determine a number $J_{p0}$ and then set the working order $J_{p}$($J=J_{p}+1$) to a slightly larger number for undersmoothing; if $P_{J}(v)$ is composed of third order B-splines,
our numerical experiments indicate that the TNS estimator is quite robust to the
number of inner knots $J_{k}$($J=J_{k}+4$).

\section{Asymptotic Properties}

This section establishes the asymptotic properties of the TNS estimator
in (\ref{est}). For easy notations, let $W=(X^{\prime },R,P_{J}(V)^{\prime })^{\prime }$
and $\theta (\tau )=(\beta (\tau )^{\prime },\alpha (%
\tau),\delta(\tau)^{\prime })^{\prime }$. We make several regularity assumptions below.
\setcounter{assumption}{0}
\renewcommand{\theassumption}{A1}
\begin{assumption}\textit{$\{Y_{i},X_{i},R_{i},Z_{i}\}_{i=1}^{n}$ is a random
sample drawn from the distribution of the pair $(Y,X,R,Z)$ generated
according to model (\ref{censor2})-(\ref{endo2}) with Assumption ID2.}
\end{assumption}
\renewcommand{\theassumption}{A2}
\begin{assumption} \textit{For all $(x,r,v)\in \mathcal{\bar{X}\times \bar{R}%
		\times \bar{V}}$, $f_{Y^{\ast }|X,R,V}(y|x,r,v)$ is bounded from above by $%
	\bar{f}$, and $f_{Y^{\ast }|X,R,V}(y|x,r,v)$ is bounded away from zero and
	differentiable in the neighborhood of $Q_{Y^{\ast }|X,R,V}(\tau |x,r,v)$
	with $f_{Y^{\ast }|X,R,V}^{\prime }(Q_{Y^{\ast }|X,R,V}(\tau |x,r,v)|x,r,v)$
	bounded in absolute value by $\bar{f}^{\prime }$ uniformly in $\tau $.}
\end{assumption}
\renewcommand{\theassumption}{A3}
\begin{assumption}\textit{$\beta (\tau )$, $\alpha (\tau )$ and $m(v,\tau )$
	are differentiable with respect to $\tau $ and $v$, and all partial
	derivatives are uniformly bounded.}
\end{assumption}
\renewcommand{\theassumption}{A4}
\begin{assumption}\textit{The control variable estimator $\hat{v}(r,z)$ satisfies 
	\begin{equation*}
		\sup_{(r,z)\in \mathcal{\bar{R}\times \bar{Z}}}\left\vert \hat{v}%
		(r,z)-v(r,z)\right\vert =O_{p}\left( \pi _{1n}\right) =o_{p}(n^{-1/4})
	\end{equation*}%
	and has an asymptotic representation,%
	\begin{equation*}
		\sqrt{n}\left( \hat{v}(r,z)-v(r,z)\right) =\frac{1}{\sqrt{n}}%
		\sum_{j=1}^{n}\phi (r,z,R_{j},Z_{j})+o_{p}(1),
	\end{equation*}%
	uniformly over $(r,z)\in \mathcal{\bar{R}\times \bar{Z}}$ with $E(\phi
	(r,z,R,Z))=0$.}
\end{assumption}
\renewcommand{\theassumption}{A5}
\begin{assumption}\textit{(i)Let $\Omega (\tau )=E[ TD_{\tau
	}(\tilde{\eta} )f_{Y^{\ast }|X,R,V}(Q_{Y^{\ast }|X,R,V}(\tau
	|X,R,V)|X,R,V)WW^{\prime }] $.%
	\begin{equation*}
		0<\underline{c}_{\Omega }\leq \text{mineig}\Omega (\tau )\leq \text{maxeig}%
		\Omega (\tau )\leq \bar{c}_{\Omega }<\infty
	\end{equation*}%
	uniformly in $J$ and $\tau $.}
	
	\textit{(ii) Let $G_{1}(\tau )=TD_{\tau }(\tilde{\eta} )(1\{Y<X^{\prime }\beta (\tau
	)-R\alpha (\tau )-m(V,\tau )\}-\tau )W$ and $G_{2}(\tau )=g(R,Z,\tau )$ with 
	$g(r,z,\tau )=E[TD_{\tau }(\tilde{\eta} )f_{Y^{\ast }|X,R,V}(Q_{Y^{\ast
		}|X,R,V}(\tau |X,R,V)|X,R,V)Wm^{\prime }(V,\tau )\phi (R,Z,r,z)]$. The
	matrix $\Sigma (\tau )=var\left[ G_{1}(\tau )+G_{2}(\tau )\right] $ satisfies%
	\begin{equation*}
		0<\underline{c}_{\Sigma }\leq \text{mineig}\Sigma (\tau )\leq \text{maxeig}%
		\Sigma (\tau )\leq \bar{c}_{\Sigma }<\infty
	\end{equation*}%
	uniformly in $J$ and $\tau $.}
\end{assumption}
\renewcommand{\theassumption}{A6}
\begin{assumption} \textit{For a constant $\kappa $ independent of $n$, $\sup_{\tau \in \lbrack \tau _{L},\tau
		_{U}]}\sup_{v\in \mathcal{\bar{V}}}\left\vert r_{J}(v,\tau )\right\vert
	\lesssim J^{-\kappa }$.}
\end{assumption}
\renewcommand{\theassumption}{A7}
\begin{assumption} \textit{Let $\zeta _{0}=\sup_{v\in \mathcal{\bar{V}}%
	}\left\Vert P_{J}(v)\right\Vert $. (i) $\tilde{\eta}_{n}\rightarrow \tilde{\eta} $, $%
	L\rightarrow \infty $, $\tilde{\eta}_{n}\left( \frac{\zeta _{0}^{2}J\ln n}{n}%
	\right) ^{-1/2}\rightarrow \infty $, $\tilde{\eta}_{n}\pi
	_{1n}^{-1}\rightarrow \infty $, $\tilde{\eta}_{n}L\rightarrow \infty $, $%
	\zeta _{0}^{2}J\left. \ln n\right/ n=o(1)$ and $nJ^{-2\kappa }=o(1).$}
	
	\textit{(ii) $J^{3}\zeta _{0}^{2}\ln ^{3}n/n=o(1)$, $nJ^{1-2\kappa }=o(1),$ $J(%
	\tilde{\eta}_{n}-\tilde{\eta} )\ln n=o(1)$, $J\pi _{1n}\ln n=o(1)$ and $J^{1-\kappa
	}\ln n=o(1).$}
\end{assumption}
Assumption A2 imposes standard conditions for quantile regressions. $f_{Y^{\ast }|X,R,V}
(y|x,r,v)$ is bounded and $f_{Y^{\ast
	}|X,R,V}(Q_{Y^{\ast }|X,R,V}(\tau |x,r,v)|x,r,v)$ is smoothness over the
bounded support $\mathcal{\bar{X}\times \bar{R}\times \bar{V}}$. Assumption
A3 posits that quantile coefficients and the function evolve gradually in $%
\tau $. These differentiability conditions, together with uniform
convergence rates of the control-variable estimator and the TNS estimator,
guarantee the validity of the sequential analysis.

Assumption A4 is a high-level assumption to regulate the uniform rate of
convergence and the linear asymptotic representation for the first-step
estimator. Besides of (\ref{ctrl}), any generic
conditional CDF estimator satisfying Assumption A4 can be used to estimate
the control variable. This requirement is standard for two-step estimators;
see Theorem 8.11 in \citet{newey1994large} for more details.

Assumption A5(i) imposes a nonsingularity condition for the series quantile
regression in the second-step estimation. Assumption A5(ii) is the local
identification condition for the TNS estimator. In particular, $G_{1}(\tau
)$ is the asymptotic influence of the series quantile estimation and $%
G_{2}(\tau )$ reflects that from estimation of the control variable.
Assumption A6 is a restriction on the approximation error that is common in
sieve estimation literature.  \citet{chen2007large} and \citet{belloni2019conditional} discussed
primitive conditions for it. For example, $m(v,\tau )$ belongs to a H\"{o}%
lder ball $\Omega (s,v)$ uniformly in $\tau $, where\ $s$ is the
smoothness index related to $\kappa $.

Assumption A7(i) constrains the tuning parameter $\tilde{\eta}_{n}$ in the
subsample selector, the number of series terms $J$ and the number of
quantile grid points $L$ for consistency. Once the basis function is
specified, the order of $\zeta _{0}$ is known---e.g, $\zeta _{0}=O(J)$ if $%
P_{J}(v)$ is a power series, and $\zeta _{0}=O(J^{1/2})$ if $P_{J}(v)$ is a
spline basis. It requires that $\tilde{\eta}_{n}$ exceeds the uniform
convergence rates of $\hat{m}(v,\tau )$ and $\hat{v}(r,z)$, and that $%
L\rightarrow \infty $ at a rate faster than $1/\tilde{\eta}_{n}$. Assumption
A7(ii) imposes a sharper restriction to achieve uniform asymptotic
representation and asymptotic normality of the TNS estimator.

\begin{theorem}
	\label{Thm3} Let Assumption A1-A7(i) hold. Then%
	\begin{equation*}
		\sup_{\tau \in \lbrack \tau _{L},\tau _{U}]}\left\Vert \hat{\theta}(\tau
		)-\theta (\tau )\right\Vert =O_{p}\left( \sqrt{\frac{J\ln n}{n}}\right)
	\end{equation*}%
	and%
	\begin{equation*}
		\sup_{v\in \mathcal{\bar{V}}}\sup_{\tau \in \lbrack \tau _{L},\tau _{U}]}|%
		\hat{m}(v,\tau )-m(v,\tau )|=O_{p}\left( \zeta _{0}\sqrt{\frac{J\ln n}{n}}%
		\right) +O_{p}(J^{-\kappa }).
	\end{equation*}
\end{theorem}

Theorem \ref{Thm3} derives uniform convergence rates for $\hat{\theta}(\tau
) $ and $\hat{m}(v,\tau )$, which is important to determine the tuning
parameter for the subsample selector. Theorem \ref{Thm4} establishes the uniform
asymptotic representation of $\hat{\theta}(\tau )$ for $\tau \in \lbrack \tau _{L},
\tau _{U}]$, acting as a foundation for deriving the asymptotic normality of 
$(\hat{\beta}(\tau ),\hat{\alpha}(\tau ))$ and the weighted bootstrap procedure. 

\begin{theorem}
	\label{Thm4}Let Assumption A1-A7 hold. Then 
	\begin{equation*}
		\sqrt{n}(\hat{\theta}(\tau )-\theta (\tau ))=\frac{1}{\sqrt{n}}%
		\sum_{i=1}^{n}\xi _{i}(\tau )+o_{p}(1)
	\end{equation*}%
	uniformly in $\tau \in \lbrack \tau _{L},\tau _{U}]$, where 
    \begin{align*}
        \xi _{i}(\tau)=-\Omega (\tau )^{-1}G_{1i}(\tau )-\Omega (\tau )^{-1}G_{2i}(\tau )
    \end{align*} 
    with %
	\begin{align*}
	    G_{1i}(\tau) &=T_{i}D_{\tau i}(\tilde{\eta} )(1\{Y_{i}<X_{i}^{\prime }\beta (\tau
	)-R_{i}\alpha (\tau )-m(V_{i},\tau )\}-\tau )W_{i}\\
    G_{2i}(\tau
	)&=g(R_{i},Z_{i},\tau )=E[TD_{\tau }(\tilde{\eta} )f_{Y^{\ast }|X,R,V}(Q_{Y^{\ast }|X,R,V}(\tau
	|X,R,V)|X,R,V)Wm^{\prime }(V,\tau )\\
    &\times\phi (R,Z,R_{i},Z_{i})|R_{i},Z_{i}]
	\end{align*}
\end{theorem}

Define two matrices, $S_{1}=(I_{K_{X}+1},0_{(K_{X}+1)\times J})$ and $S_{2}=(0_{J\times
	(K_{X}+1)},I_{J})$, which extract, respectively, the first $(K_{X}+1)$
elements and the last $J$ elements of a $(K_{X}+1+J)\times 1$ vector. Theorem
\ref{Thm5} demonstrates the $\sqrt{n}$-asymptotic normality.
\begin{theorem}
	\label{Thm5} Let Assumption A1-A7 hold, then $\sqrt{n}\left( \dbinom{\hat{%
			\beta}(\tau )}{\hat{\alpha}(\tau )}-\dbinom{\beta (\tau )}{\alpha (\tau )}%
	\right) $ converges to a mean zero Gaussian process with the covariance
	function $\Sigma _{\beta \alpha }(\tau _{1},\tau _{2})=E\left[ S_{1}\xi
	_{i}(\tau _{1})\xi _{i}(\tau _{2})^{\prime }S_{1}^{\prime }\right] $.
\end{theorem}

For statistical inference, the analytical estimator of the
variance-covariance function is straightforward and expressed as 
\begin{equation*}
	\hat{\Sigma}_{\beta \alpha }(\tau _{1},\tau _{2})=\frac{1}{n}%
	\sum_{i=1}^{n}S_{1}\hat{\xi}_{i}(\tau _{1})\hat{\xi}_{i}(\tau _{2})^{\prime
	}S_{1}^{\prime },
\end{equation*}%
where $\hat{\xi}_{i}(\tau )=-\hat{\Omega}(\tau )^{-1}\hat{G}_{1i}(\tau )-%
\hat{\Omega}(\tau )^{-1}\hat{G}_{2i}(\tau )$ with 
\begin{eqnarray*}
	\hat{\Omega}(\tau ) &=&\frac{1}{nh_{n}}\sum_{i=1}^{n}T_{i}\hat{\hat{D}}%
	_{\tau i}1\{0\leq \hat{\varepsilon}_{\tau i}\leq h_{n}\}\hat{W}_{i}\hat{W}%
	_{i}^{\prime }, \\
	\hat{G}_{1i}(\tau ) &=&T_{i}\hat{\hat{D}}_{\tau i}(1\{Y_{i}<X_{i}^{\prime }%
	\hat{\beta}(\tau )-R_{i}\hat{\alpha}(\tau )-\hat{m}(\hat{V}_{i},\tau
	)\}-\tau )\hat{W}_{i}, \\
	\hat{G}_{2i}(\tau ) &=&\frac{1}{nh_{n}}\sum_{j=1}^{n}T_{j}\hat{\hat{D}}%
	_{\tau j}1\{0\leq \hat{\varepsilon}_{\tau j}\leq h_{n}\}\hat{W}_{j}\hat{m}%
	_{1}^{\prime }(\hat{V}_{j},\tau )\phi (R_{j},Z_{j},R_{i},Z_{i}),
\end{eqnarray*}%
where $\hat{m}_{1}^{\prime }(v,\tau )=\frac{\partial \hat{m}(v,\tau )}{%
	\partial v}$, $\hat{\hat{D}}_{\tau i}=1\{X_{i}^{\prime }\hat{\beta}(\tau
)+R_{i}\hat{\alpha}(\tau )+\hat{m}(\hat{V}_{i},\tau )>\tilde{\eta}_{n}\}$
and $\hat{\varepsilon}_{\tau i}=Y_{i}-X_{i}^{\prime }\hat{\beta}(\tau )-R_{i}%
\hat{\alpha}(\tau )-\hat{m}(\hat{V}_{i},\tau )$ as $J^{2}/nh_{n}\rightarrow
0 $ and $h_{n}\rightarrow 0$.

It is well known that analytical inference is often impractical for applied
researchers, since the estimation of the variance--covariance function
entails numerous complex components, such as the estimated derivatives of
unknown functions and the bandwidth selection required for density
estimation. As an alternative, we offer a weighted bootstrap procedure for
the TNS estimator. Specifically, generate an i.i.d sample $\{\varsigma
_{i}\}_{i=1}^{n}$ from a nonnegative random variable $\varsigma $, which is
independent of data with $E(\varsigma )=var(\varsigma )=1$, for example, the
standard exponential random variable. In the $t$-th draw, the bootstrap
version of the TNS estimator is obtained by%
\begin{equation*}
	(\hat{\beta}^{\dag }(\tau ),\hat{\alpha}^{\dag }(\tau ),\hat{\delta}^{\dag
	}(\tau ))=\underset{(b,a,d)\in \mathcal{B\times A\times D}}{\arg \min }\frac{%
		1}{n}\sum_{i=1}^{n}\varsigma _{i}T_{i}\hat{D}_{\tau i}\rho _{\tau
	}(Y_{i}-X_{i}^{\prime }b-R_{i}a-P_{J}(\hat{V}_{i})^{\prime }d).
\end{equation*}
\renewcommand{\theassumption}{A8}
\begin{assumption} \textit{The weights $\{\varsigma _{i}\}_{i=1}^{n}$ are i.i.d
draws from a nonnegative random variable $\varsigma $ independent of data
such that $E(\varsigma )=var(\varsigma )=1$ and $E(\varsigma ^{2+\pi
})<\infty $ for some $\pi >0$.}
\end{assumption}

\begin{theorem}
	\label{Thm6} Let Assumption A1-A8 hold, then $\sqrt{n}\left( \dbinom{\hat{\beta}^{\dag }(\tau )}{\hat{\alpha}^{\dag }(\tau )}-\dbinom{\hat{\beta}(\tau
		)}{\hat{\alpha}(\tau )}\right) $ converges to a mean zero Gaussian process
	with the covariance function as $\Sigma _{\beta \alpha }(\tau _{1},\tau
	_{2})=E\left[ S_{1}\xi _{i}(\tau _{1})\xi _{i}(\tau _{2})^{\prime
	}S_{1}^{\prime }\right] $.
\end{theorem}

Theorem \ref{Thm6} implies that the asymptotic distribution of $\sqrt{n}%
\left( (\hat{\beta}(\tau )^\prime,\hat{\alpha}(\tau ))^\prime-(\beta (\tau )^\prime,%
\alpha (\tau ))^\prime\right) $ is well approximated by the empirical
distributions of $\sqrt{n}\left( (\hat{\beta}_{t}^{\dag }(\tau ),%
	\hat{\alpha}_{t}^{\dag }(\tau ))^\prime-(\hat{\beta}(\tau )^\prime,\hat{\alpha}(\tau ))^\prime\right) $, $t=1,\cdots ,T$. Subsequently, the standard errors and
confidence intervals for the TNS estimator can be computed using the
bootstrap draws.

\section{Extension: Non-additive Structure}

Section 2 focuses on a widely employed
semiparametric specification. This section
extends TNS estimation to the general case of non-additive structures as (%
\ref{censor1})-(\ref{endo1}) for completeness.%
\begin{eqnarray*}
	Y &=&\max \{Y^{\ast },0\}=\max \{m_{Y^{\ast }}(X,R,V,U),0\}, \\
	R &=&m_{R}(Z,V).
\end{eqnarray*}%
Compared with model (\ref{censor2})-(\ref{endo2}), the fully nonparametric
model allows for more general specifications. For example,
\begin{example}\label{ex4}
\begin{eqnarray*}
	Y &=&\max \{X^{\prime }\beta +R\alpha +\sqrt{X^{\prime }\tilde{\beta}+R%
		\tilde{\alpha}}\varepsilon ^{\ast },0\}, \\
	R &=&Z^{\prime }\gamma +\sqrt{Z^{\prime }\gamma }\epsilon ^{\ast }
\end{eqnarray*}%
where $\epsilon ^{\ast }=\Phi ^{-1}(V)$ with $V\sim U(0,1)$, and $%
\varepsilon ^{\ast }=\rho \Phi ^{-1}(V)+\sqrt{1-\rho ^{2}}\Phi ^{-1}(U)$
with $U\sim U(0,1)$. 
\end{example}
\begin{example}\label{ex3}
\begin{eqnarray*}
	Y &=&\max \{X^{\prime }\beta +R\alpha +\sigma _{Y^{\ast }}(X,R)\varepsilon
	^{\ast },0\}, \\
	R &=&Z^{\prime }\gamma +\sigma _{^{R}}(Z)\epsilon ^{\ast }
\end{eqnarray*}%
with $\sigma _{Y^{\ast }}(x,r)>0$ and $\sigma _{^{R}}(z)>0$ for all $(x,r,z)$%
, where $\epsilon ^{\ast }=\Phi ^{-1}(V)$ with $V\sim U(0,1)$ and $%
\varepsilon ^{\ast }=\rho \Phi ^{-1}(V)+\sqrt{1-\rho ^{2}}\Phi ^{-1}(U)$
with $U\sim U(0,1)$.
\end{example}

The identification result is a straightforward extension of Theorem \ref{Thm2}.
\begin{theorem}
	\label{Thm7} Under the model (\ref{censor1})-(\ref{endo1}), if Assumption ID1 
	and the following conditions hold,\\
	(i) (Initial Condition) There exists $\tau _{0}\in (0,1)$
	such that $\Pr (Q_{Y^{\ast }|X,R,V}(\tau |X,R,V)>0)=1$ for any $\tau
	_{0}\leq \tau \leq \tau _{U}$.\\
	(ii) (Continuity) $m_{Y^{\ast }}(x,r,v,\tau )$ is
	Lipschitz continuous in $\tau $ for all $(x,r,v)$.\\
	Then $\nabla _{r}m_{Y^{\ast }}(x,r,v,\tau )$ is identified for any $\tau \in
	\lbrack \tau _{L},\tau _{U}]$.
\end{theorem}

Next, the TNS estimation procedure for model (\ref{censor1})-(\ref%
{endo1}) is as follows:

\medskip
\noindent\textit{Step 1: control variable estimation}
\medskip

Any generic estimator for the conditional CDF of $R$ given $%
Z $ can be employed as the first-step procedure. \citet{cattaneo2024boundary}
proposed a boundary adaptive method to estimate $V_{i}=F_{R|Z}(R_{i}|Z_{i})$
by $\hat{V}_{i}=\hat{v}(R_{i},Z_{i})$, where%
\begin{equation*}
	\hat{v}(r,z)=\arg \min_{c_{r}}\sum_{j=1}^{n}\left( \hat{v}_{z}(R_{j},z)-%
	\dbinom{1}{R_{j}-r}^{\prime }c_{r}\right) ^{2}k_{h_{r}}(R_{j}-r).
\end{equation*}%
Here, $k_{h_{r}}(R_{j}-r)=\frac{1}{h_{r}}k\left( \frac{R_{i}-r}{h_{r}}%
\right) $ and $\hat{v}_{z}(r,z)$ is the minimizer of
\begin{eqnarray*}
	\sum_{i=1}^{n}\left(
	1\{R_{i}<r\}-\dbinom{1}{Z_{i}-z}^{\prime }c_{z}\right) ^{2}\frac{%
		1}{h_{z}^{K_{Z}}}\prod_{l=1}^{K_{Z}}k_{l}\left( \frac{Z_{li}-z_{l}}{h_{z}}%
	\right).
\end{eqnarray*} 
Alternatively, Chernozhukov et al. (2013) proposed a\ QR-CDF method,%
\begin{equation*}
	\hat{V}_{i}=\int_{\omega }^{1-\omega }1\{\hat{Q}_{R|Z}(v|Z_{i})<R_{i}\}dv,
\end{equation*}%
where $\hat{Q}_{R|Z}(v|z)$ is a nonparametric estimator obtained via a
series quantile regression. To mitigate inaccuracies at extreme quantiles,
the trimming cut-off $\omega $ is set to converge to zero at a certain rate.\footnote{Here the QR-CDF estimator can be used because $Q_{Y^{\ast }|X,R,V}(\tau |X,R,V)$ is fully nonparametric and its estimator has a slow rate of convergence.}

\medskip
\noindent\textit{Step 2: distributional effects estimation}
\medskip

Write $S=(X^{\prime },R,V^{\prime })^{\prime }$ and $\hat{S}=(X^{\prime },R,%
\hat{V}^{\prime })^{\prime }$. For notational convenience, $m_{Y^{\ast
}}(s,\tau )$ is shorthand for $m_{Y^{\ast }}(x,r,v,\tau )$. In the
second-step estimation, define $P_{J_{s}}(s)=(p_{1}(s),\cdots
,p_{J_{s}}(s))^{\prime }$ as a $J_{s}\times 1$ vector of basis functions
with a vector of coefficients $\delta _{s}(\tau )$, so that $%
P_{J_{s}}(s)^{\prime }\delta _{s}(\tau )$ approximates $m_{Y^{\ast }}(s,\tau
)$ well in the sense that the approximation error is asymptotically
negligible as $J_{s}\rightarrow \infty $. For any $\tau \in \lbrack \tau
_{0},\tau _{U}]$, estimate $m_{Y^{\ast }}(s,\tau )$ by 
\begin{equation*}
	\hat{m}_{Y^{\ast }}(s,\tau )=P_{J_{s}}(s)^{\prime }\hat{\delta}_{s}(\tau ),
\end{equation*}%
where 
\begin{equation*}
	\hat{\delta}_{s}(\tau )=\underset{d\in \mathcal{D}}{\min }\frac{1}{n}%
	\sum_{i=1}^{n}T_{i}\rho _{\tau }(Y_{i}-P_{J_{s}}(\hat{S}_{i})^{\prime }d).
\end{equation*}
For $[\tau _{L},\tau _{0})$, choose a sequence of grids $\tau _{L}<\cdots
<\tau _{1}<\tau _{0}$ that partitions the interval, and estimate $\delta
_{s}(\tau _{l})$ sequentially for $l=1,\cdots ,L$ by%
\begin{equation*}
	\hat{\delta}_{s}(\tau _{l})=\underset{d\in \mathcal{D}}{\arg \min }\frac{1}{n%
	}\sum_{i=1}^{n}T_{i}1\{P_{J_{s}}(\hat{S}_{i})^{\prime }\hat{\delta}_{s}(\tau
	_{l-1})>\tilde{\eta}_{n}\}\rho _{\tau _{l}}(Y_{i}-P_{J_{s}}(\hat{S}%
	_{i})^{\prime }d)
\end{equation*}%
where the feasible selector $1\{P_{J_{s}}(\hat{S}_{i})^{\prime }\hat{\delta}_{s}(\tau _{l-1})>%
\tilde{\eta}_{n}\}$ restricts the sample to observations with $P_{J_{s}}(%
\hat{S}_{i})^{\prime }\hat{\delta}_{s}(\tau _{l-1})$ exceeding $\tilde{\eta}%
_{n}$.

In general, for any $\tau \in \lbrack \tau _{L},\tau _{0})$, there
exists an $l$ such that $\tau \in \lbrack \tau _{l+1},\tau _{l})$, and the
TNS estimator of $\delta _{s}(\tau )$ is defined as%
\begin{equation*}
	\hat{\delta}_{s}(\tau )=\underset{d\in \mathcal{D}}{\arg \min }\frac{1}{n}%
	\sum_{i=1}^{n}T_{i}1\{P_{J_{s}}(\hat{S}_{i})^{\prime }\hat{\delta}_{s}(\tau
	_{l})>\tilde{\eta}_{n}\}\rho _{\tau }(Y_{i}-P_{J_{s}}(\hat{S}_{i})^{\prime
	}d).
\end{equation*}%
Accordingly, the consistent estimator of $m_{Y^{\ast }}(s,\tau )$ is%
\begin{equation*}
	\hat{m}_{Y^{\ast }}(s,\tau )=P_{J_{s}}(s)^{\prime }\hat{\delta}_{s}(\tau ).
\end{equation*}%
The asymptotic properties of $\hat{\delta}_{s}(\tau )$ and $\hat{m}_{Y^{\ast
}}(s,\tau )$ are omitted, as they are analogous to Theorem \ref{Thm3}-\ref%
{Thm5}, except that no coefficient estimator attains $\sqrt{n}$ asymptotic
normality.

There is a trade-off between model (\ref{censor1})-(\ref{endo1}) and model (%
\ref{censor2})-(\ref{endo2}). Without structural restrictions, model (\ref%
{censor1})-(\ref{endo1}) is in favor of generality, nevertheless, the
resulting estimator suffers from the curse of dimensionality. Also, the
first-step estimation has practical difficulties of choosing multiple
bandwidths or carefully setting an asymptotically negligible trimming
cut-off $\omega $. By contrast, model (\ref{censor2})-(\ref{endo2}) is more
tractable since quantile coefficients directly characterize heterogeneous
distributional effects of interest, and the semiparametric estimator is $%
\sqrt{n}$-asymptotic normality. In practice, an appropriate semiparametric
specification is often recommended as a reasonable and tractable compromise.

\section{Monte Carlo Simulation}

To illustrate the finite sample performance of TNS estimator, we consider
two data generating processes (DGPs) with the model specification as%
\begin{eqnarray*}
	Y &=&\max \{Y^{\ast },0\}=\max \{\beta _{0}+\alpha R+\beta _{1}X+\varepsilon
	,0\}, \\
	R &=&\gamma _{0}+\gamma _{1}Z_{1}+\gamma _{2}X+\epsilon .
\end{eqnarray*}

\noindent \textbf{DGP I.} $X$ and $Z_{1}$ are mutually independent standard
normal random variables. Set $\gamma _{0}=\gamma _{1}=\gamma _{2}=\alpha
=\beta _{1}=1$ and $\beta _{0}$ to be the value such that the censoring
percentage is 30\%. The error term in the endogenous regressor equation is
homoskedastic---$\epsilon =\Phi ^{-1}(V)$, and the error term in the latent
outcome equation is heteroskedastic $\varepsilon =\frac{1}{\sqrt{2}}\rho
\Phi ^{-1}(V)+\frac{1}{\sqrt{2}}[0.4(0.5R+0.5X+2)+\sqrt{1-\rho ^{2}}]\Phi
^{-1}(U)$, where $V$ is independent of $(X,Z_{1})$ with $V\sim U(0,1)$, $U$
is independent of $(X,Z_{1},V)$ with $U\sim U(0,1)$ and $\rho =0.5$. Under
this setup, the distributional effects of $R$ and $X$ are $\alpha (\tau )=1+%
\frac{1}{5\sqrt{2}}\Phi ^{-1}(\tau )$ and $\beta _{1}(\tau )=1+\frac{1}{5%
	\sqrt{2}}\Phi ^{-1}(\tau )$.\textbf{\medskip }

\noindent \textbf{DGP II.} The data is generated same as DGP I except that
the error term in the endogenous regressor equation is heteroskedastic---$\epsilon
=0.5(2+0.5X+0.5Z_{1})\Phi ^{-1}(V)$.\textbf{\medskip }

In DGP I, the latent outcome features a heteroskedastic error implying that
the distributional effects of interest are heterogeneous across quantile
levels. In DGP II, both the endogenous regressor and the latent outcome
contain heteroskedastic errors. We implement the TNS estimation algorithm:
In the first-step, we set $(\tilde{v}_{1},\tilde{v}_{2})=(0.25,0.75)$. In
the second step, we remove 1\% observations with extreme values of $%
(X_{i},R_{i},Z_{i})$, set the grid step to $\tau _{l-1}-\tau _{l}=0.01$ with 
$\tau _{0}=0.99$ and $\tau _{L}=0.3$ (so that $L=90$), and adopt $%
(q_{0},q_{1})=(0.1,0.03)$. Besides, the induced unknown function of $V_{i}$
is approximated by a raw polynomial series with\ the degree $J_{p}=3$($%
J=J_{p}+1$) and a cubic B-spline series with the number of inner knots $%
J_{k}=3$($J=J_{k}+4$), respectively. Each experiment is replicated $G=1000$
times. The sample sizes are $n=250$, $500$, $1000$. More simulation results with
different $J_{p}$ and $J_{k}$ are provided in the supplementary materials.

First, Tables \ref{table1}-\ref{table2} report the bias (Bias) and root of mean square error
(RMSE) of TNS estimates at $(0.9,0.7,0.5,0.3)$ quantiles, which are
calculated by%
\begin{equation*}
	Bias(\tau )=\frac{1}{G}\sum_{g=1}^{G}\left( \dbinom{\hat{\alpha}_{g}(\tau )}{%
		\hat{\beta}_{1g}(\tau )}-\dbinom{\alpha (\tau )}{\beta _{1}(\tau )}\right) 
\end{equation*}
and 
\begin{equation*}
	RMSE(\tau )=\left[ \frac{1}{G}\sum_{g=1}^{G}\left( \dbinom{\hat{%
			\alpha}_{g}(\tau )}{\hat{\beta}_{1g}(\tau )}-\dbinom{\alpha (\tau )}{\beta
		_{1}(\tau )}\right) ^{2}\right] ^{1/2},
\end{equation*}
where $\dbinom{\hat{\alpha}_{g}(\tau )}{\hat{\beta}_{1g}(\tau )}$ denotes
the estimate from the $g$-th simulation replication. For comparison, we also
report the CQIV-DR estimates which are computed using Algorithm 1 of
\citet{chernozhukov2015quantile}. As shown in both tables, the TNS method
outperforms the CQIV-DR method in estimating the distributional effects,
yielding negligible bias and smaller RMSE values.
\begin{landscape}%
\begin{table}
\small
	\caption{Estimation---Polynomial series}
	\label{table1}
		\begin{tabular}{ccccccccccccccc}
				\hline
			&  & \multicolumn{2}{c}{${n=250}$} & \multicolumn{2}{c}{${n=500%
				}$} & \multicolumn{2}{c}{${n=1000}$} &  & \multicolumn{2}{c}{${
					n=250}$} & \multicolumn{2}{c}{${n=500}$} & \multicolumn{2}{c}{$%
				{n=1000}$} \\ 
			&  & {Bias} & {RMSE} & {Bias} & {RMSE} & {
				Bias} & {RMSE} &  & {Bias} & {RMSE} & {Bias} & 
			{RMSE} & {Bias} & {RMSE} \\ \cline{3-15}
			{TNS} &  & \multicolumn{6}{c}{DGP I} &  & \multicolumn{6}{c}%
			{DGP II} \\ 
			${\tau =0.9}$ & ${\alpha (\tau )}$ & {0.008} & {0.130} & {0.005} & {0.090} & {0.008} & {0.062} & & {-0.007} & {0.140} & {0.003} & {0.098} & {0.001} & {0.069} \\ 
			& ${\beta_{1}(\tau)}$ & {0.012} & {0.177} & 
			{0.003} & {0.117} & {-0.006} & {0.084} & & 
			{-0.038} & {0.201} & {-0.043} & {0.151} & 
			{-0.044} & {0.105} \\ [6pt]
			${\tau =0.7}$ & ${\alpha (\tau )}$ & {0.006} & {
			0.106} & {0.004} & {0.074} & {0.007} & {0.053} & 
			& {0.023} & {0.114} & {0.025} & {0.087} & 
			{0.028} & {0.064} \\ 
			& ${\beta }_{1}{(\tau )}$ & {0.013} & {0.146} & 
			{0.003} & {0.097} & {-0.002} & {0.071} &  & 
			{-0.012} & {0.167} & {-0.017} & {0.116} & 
			{-0.015} & {0.082} \\ [6pt]
			${\tau =0.5}$ & ${\alpha (\tau )}$ & {0.008} & {
			0.107} & {0.006} & {0.076} & {0.008} & {0.053} & 
			& {0.043} & {0.122} & {0.046} & {0.094} & 
			{0.044} & {0.071} \\ 
			& ${\beta }_{1}{(\tau )}$ & {0.007} & {0.146} & 
			{0.005} & {0.104} & {-0.004} & {0.069} &  & 
			{0.007} & {0.168} & {0.002} & {0.116} & {
			0.002} & {0.078} \\ [6pt]
			${\tau =0.3}$ & ${\alpha (\tau )}$ & {0.012} & {
			0.121} & {0.007} & {0.082} & {0.010} & {0.057} & 
			& {0.068} & {0.149} & {0.061} & {0.111} & 
			{0.062} & {0.089} \\ 
			& ${\beta }_{1}{(\tau )}$ & {0.009} & {0.162} & 
			{0.005} & {0.110} & {-0.005} & {0.074} &  & 
			{0.032} & {0.196} & {0.024} & {0.131} & {
			0.020} & {0.093} \\ \hline
			{CQIV} &  & \multicolumn{6}{c}{DGP I} &  & \multicolumn{6}{c}%
			{DGP II} \\ 
			${\tau =0.9}$ & ${\alpha (\tau )}$ & {-0.383} & {
			0.419} & {-0.382} & {0.404} & {-0.390} & {0.403} &  & {
			-0.225} & {0.270} & {-0.223} & {0.249} & {-0.224} & {0.238}
			\\ 
			& ${\small \beta }_{1}{\small (\tau )}$ & {-0.366} & {0.427} & 
			{-0.372} & {0.408} & {-0.377} & {0.395} &  & {-0.490} & 
			{0.530} & {-0.512} & {0.537} & {-0.539} & {0.552} \\  [6pt]
			${\tau =0.7}$ & ${\alpha (\tau )}$ & {-0.298} & {
			0.329} & {-0.296} & {0.313} & {-0.302} & {0.313} &  & {
			-0.190} & {0.226} & {-0.191} & {0.211} & {-0.195} & {0.207}
			\\ 
			& ${\beta }_{1}{(\tau )}$ & {-0.284} & {0.325} & 
			{-0.279} & {0.305} & {-0.291} & {0.305} &  & {-0.357} & 
			{0.390} & {-0.371} & {0.393} & {-0.391} & {0.403} \\ [6pt]
			${\tau =0.5}$ & ${\alpha (\tau )}$ & {-0.254} & {
			0.283} & {-0.250} & {0.267} & {-0.262} & {0.273} &  & {
			-0.168} & {0.204} & {-0.173} & {0.195} & {-0.181} & {0.194}
			\\ 
			& ${\beta }_{1}{(\tau )}$ & {-0.248} & {0.287} & 
			{-0.247} & {0.271} & {-0.252} & {0.266} &  & {-0.298} & 
			{0.334} & {-0.304} & {0.325} & {-0.321} & {0.334} \\  [6pt]
			${\tau =0.3}$ & ${\alpha (\tau )}$ & {-0.233} & {
			0.267} & {-0.234} & {0.254} & {-0.243} & {0.256} &  & {
			-0.157} & {0.199} & {-0.166} & {0.192} & {-0.175} & {0.190}
			\\ 
			& ${\beta }_{1}{(\tau )}$ & {-0.237} & {0.280} & 
			{-0.234} & {0.261} & {-0.239} & {0.255} &  & {-0.263} & 
			{0.305} & {-0.275} & {0.299} & {-0.288} & {0.302} \\ \hline
		\end{tabular}%
	\end{table}
\begin{table}
	\small
	\caption{Estimation---B-spline series}
	\label{table2}
	\begin{tabular}{ccccccccccccccc}
		\hline
		&  & \multicolumn{2}{c}{${\small n=250}$} & \multicolumn{2}{c}{${\small n=500%
			}$} & \multicolumn{2}{c}{${\small n=1000}$} &  & \multicolumn{2}{c}{${\small %
				n=250}$} & \multicolumn{2}{c}{${\small n=500}$} & \multicolumn{2}{c}{$%
			{\small n=1000}$} \\ 
		&  & {\small Bias} & {\small RMSE} & {\small Bias} & {\small RMSE} & {\small %
			Bias} & {\small RMSE} &  & {\small Bias} & {\small RMSE} & {\small Bias} & 
		{\small RMSE} & {\small Bias} & {\small RMSE} \\ \cline{3-15}
		{\small SCFCQ} &  & \multicolumn{6}{c}{\small DGP I} &  & \multicolumn{6}{c}%
		{\small DGP II} \\ 
		${\small \tau =0.9}$ & ${\small \alpha (\tau )}$ & {\small 0.010} & {\small %
			0.134} & {\small 0.004} & {\small 0.086} & {\small 0.003} & {\small 0.059} & 
		\multicolumn{1}{r}{} & {\small -0.005} & {\small 0.141} & {\small -0.001} & 
		{\small 0.096} & {\small -0.004} & {\small 0.069} \\ 
		& ${\small \beta }_{1}{\small (\tau )}$ & {\small -0.000} & {\small 0.176} & 
		{\small 0.006} & {\small 0.120} & {\small 0.003} & {\small 0.081} & 
		\multicolumn{1}{r}{} & {\small -0.031} & {\small 0.196} & {\small -0.042} & 
		{\small 0.146} & {\small -0.041} & {\small 0.104} \\ 
		&  &  &  &  &  &  &  &  &  &  &  &  &  &  \\ 
		${\small \tau =0.7}$ & ${\small \alpha (\tau )}$ & {\small 0.008} & {\small %
			0.110} & {\small 0.004} & {\small 0.071} & {\small 0.002} & {\small 0.051} & 
		\multicolumn{1}{r}{} & {\small 0.031} & {\small 0.115} & {\small 0.030} & 
		{\small 0.084} & {\small 0.025} & {\small 0.060} \\ 
		& ${\small \beta }_{1}{\small (\tau )}$ & {\small 0.007} & {\small 0.148} & 
		{\small 0.005} & {\small 0.099} & {\small 0.002} & {\small 0.072} & 
		\multicolumn{1}{r}{} & {\small -0.000} & {\small 0.161} & {\small -0.010} & 
		{\small 0.116} & {\small -0.012} & {\small 0.078} \\ 
		&  &  &  &  &  &  &  &  &  &  &  &  &  &  \\ 
		${\small \tau =0.5}$ & ${\small \alpha (\tau )}$ & {\small 0.013} & {\small %
			0.107} & {\small 0.004} & {\small 0.073} & {\small 0.004} & {\small 0.052} & 
		\multicolumn{1}{r}{} & {\small 0.052} & {\small 0.122} & {\small 0.048} & 
		{\small 0.095} & {\small 0.044} & {\small 0.070} \\ 
		& ${\small \beta }_{1}{\small (\tau )}$ & {\small 0.004} & {\small 0.145} & 
		{\small 0.003} & {\small 0.097} & {\small 0.003} & {\small 0.071} & 
		\multicolumn{1}{r}{} & {\small 0.015} & {\small 0.170} & {\small 0.010} & 
		{\small 0.119} & {\small 0.004} & {\small 0.079} \\ 
		&  &  &  &  &  &  &  &  &  &  &  &  &  &  \\ 
		${\small \tau =0.3}$ & ${\small \alpha (\tau )}$ & {\small 0.013} & {\small %
			0.115} & {\small 0.006} & {\small 0.083} & {\small 0.005} & {\small 0.056} & 
		\multicolumn{1}{r}{} & {\small 0.068} & {\small 0.149} & {\small 0.064} & 
		{\small 0.115} & {\small 0.062} & {\small 0.089} \\ 
		& ${\small \beta }_{1}{\small (\tau )}$ & {\small 0.009} & {\small 0.158} & 
		{\small 0.007} & {\small 0.106} & {\small 0.003} & {\small 0.076} & 
		\multicolumn{1}{r}{} & {\small 0.044} & {\small 0.201} & {\small 0.028} & 
		{\small 0.129} & {\small 0.021} & {\small 0.094} \\ \hline
		{\small CQIV} &  & \multicolumn{6}{c}{\small DGP I} &  & \multicolumn{6}{c}%
		{\small DGP II} \\ 
		${\small \tau =0.9}$ & ${\small \alpha (\tau )}$ & {\small -0.411} & {\small %
			0.445} & -0.404 & 0.425 & {\small -0.412} & {\small 0.422} & 
		\multicolumn{1}{r}{} & {\small -0.232} & {\small 0.280} & -0.233 & 0.261 & 
		{\small -0.233} & {\small 0.249} \\ 
		& ${\small \beta }_{1}{\small (\tau )}$ & {\small -0.356} & {\small 0.421} & 
		-0.385 & 0.420 & {\small -0.395} & {\small 0.411} & \multicolumn{1}{r}{} & 
		{\small -0.496} & {\small 0.540} & -0.527 & 0.548 & {\small -0.556} & 
		{\small 0.567} \\ 
		&  &  &  &  &  &  &  &  &  &  &  &  &  &  \\ 
		${\small \tau =0.7}$ & ${\small \alpha (\tau )}$ & {\small -0.318} & {\small %
			0.343} & -0.314 & 0.329 & {\small -0.327} & {\small 0.334} & 
		\multicolumn{1}{r}{} & {\small -0.203} & {\small 0.235} & -0.207 & 0.226 & 
		{\small -0.207} & {\small 0.218} \\ 
		& ${\small \beta }_{1}{\small (\tau )}$ & {\small -0.292} & {\small 0.336} & 
		-0.308 & 0.328 & {\small -0.313} & {\small 0.323} & \multicolumn{1}{r}{} & 
		{\small -0.368} & {\small 0.403} & -0.392 & 0.408 & {\small -0.409} & 
		{\small 0.418} \\ 
		&  &  &  &  &  &  &  &  &  &  &  &  &  &  \\ 
		${\small \tau =0.5}$ & ${\small \alpha (\tau )}$ & {\small -0.279} & {\small %
			0.302} & -0.275 & 0.289 & {\small -0.285} & {\small 0.292} & 
		\multicolumn{1}{r}{} & {\small -0.183} & {\small 0.216} & -0.192 & 0.212 & 
		{\small -0.195} & {\small 0.205} \\ 
		& ${\small \beta }_{1}{\small (\tau )}$ & {\small -0.255} & {\small 0.296} & 
		-0.273 & 0.292 & {\small -0.277} & {\small 0.286} & \multicolumn{1}{r}{} & 
		{\small -0.307} & {\small 0.340} & -0.329 & 0.346 & {\small -0.342} & 
		{\small 0.352} \\ 
		&  &  &  &  &  &  &  &  &  &  &  &  &  &  \\ 
		${\small \tau =0.3}$ & ${\small \alpha (\tau )}$ & {\small -0.257} & {\small %
			0.284} & -0.260 & 0.274 & {\small -0.268} & {\small 0.275} & 
		\multicolumn{1}{r}{} & {\small -0.174} & {\small 0.214} & -0.190 & 0.212 & 
		{\small -0.196} & {\small 0.208} \\ 
		& ${\small \beta }_{1}{\small (\tau )}$ & {\small -0.243} & {\small 0.284} & 
		-0.258 & 0.280 & {\small -0.264} & {\small 0.274} & \multicolumn{1}{r}{} & 
		{\small -0.281} & {\small 0.320} & -0.297 & 0.317 & {\small -0.308} & 
		{\small 0.320} \\ \hline
	\end{tabular}%
\end{table}
\end{landscape}%
Second, RMSE of the TNS estimator decreases substantially as the sample
size increases. When the sample size increases from $250$ to $1000$, they
are almost halved in accordance with the $\sqrt{n}$-asymptotic normality of $%
\dbinom{\hat{\alpha}(\tau )}{\hat{\beta}(\tau )}$ established in Theorem \ref%
{Thm5}.

Third, the unknown function is approximated by a raw polynomial series in
Table \ref{table1}, and by a B-spline series in Table \ref{table2}. The similar performance implies that the TNS
estimator is robust to basis functions.

Finally, more numerical results are available in the supplementary
materials. Together with Tables \ref{table1}-\ref{table2}, they evaluate the performance of the TNS estimator under different
values of $J$. By Table \ref{table1} ($J_{p}=3$) and the first supplementary table ($J_{p}=2$), we find that
biases generally decrease as the order increases from $2$ to $3$. Moreover,
if the order of the spline basis is fixed, the TNS method exhibits robust
performance with $J_{k}=2,3,4$.
\section{Application: Heterogeneous Income Elasticities of Demand}

The income elasticity of demand describes how household expenditure on
particular goods or services responds to household income(or total
expenditure). This concept has long attracted substantial attention for its
ability to reveal households' strategies in resource allocation.
\citet{bobonis2009allocation} emphasized that resource allocation is an endogenous
optimization process. A large body of related works have committed to
elasticity of demand on various commodities, dealing with the endogeneity
arising from income(total expenditure). For example, \citet{blundell2007semi}
evaluated UK households' income elasticities of demand for seven categories
of nondurables and services among UK households, using the household head's
earnings as an excluded instrumental variable.

Despite substantial progress in the literature, most studies on income
elasticities have focused exclusively on the average effect of income (or
total expenditure), and implicitly assumed that household preferences are
constant. The heterogeneity across relative ranks of commodity expenditure has often been
ignored, leading to an incomplete understanding of distributional
heterogeneity. \citet{lyssiotou1999preference} found that income elasticities are
inherently heterogeneous because they differ across individuals with
distinct consumption patterns. Also, \citet{lewbel2006engel}
documented that the discrepancy between observed budget share and its fitted
value widens markedly as income (total expenditure) rises. Since then,
heteroskedasticity has been recognized as an important source of
heterogeneity in the income elasticities of demand (\citet{pryce2019alcohol}).

In practice, a non-negligible proportion of households never purchase
certain goods or services. For instance, individuals who never smoke or
drink allocate zero expense to cigarettes or alcohol, resulting in censored
expenditure data at zero. Such censoring, which frequently occurs in
household consumption data, poses challenges for estimating distributional
effects with respect to income. \citet{chernozhukov2015quantile} applied the CQIV
method to estimate the elasticity of alcohol demand using data from the UK
Family Expenditure Survey. However, possible heteroskedasticity within
expenditure can induce heterogeneous distributional effects, which in turn
may lead to incorrect subsample selection and then inconsistency of the CQIV
estimator.

This section intends to estimate heterogeneous expenditure(income) elasticities
of demand on alcohol, fares, and fuel across conditional distribution of
commodity expenses, fully admitting endogeneity, heteroskedasticity and
censoring. The 1995 UK Family Expenditure Survey data shows that approximately 30\% households reported no alcohol purchases, 46\% recorded zero expenditure on fares, and 2\% reported no fuel expenditure. The dependent variables---budget
shares---are therefore subject to censoring. In many studies of income
elasticity, total expenditure is employed instead of income to mitigate
measurement error (\citet{blundell2007semi,chernozhukov2015quantile}). Accordingly, our regression specification is%
\begin{eqnarray}\label{spec1}
	\begin{aligned}
	Y &=\max \{X^{\prime }\beta (U)+R\alpha (U)+m(V,U),0\},  \label{spec1} \\
	R &=\gamma _{0}+Z_{1}\gamma _{1}+X^{^{\prime }}\gamma _{2}+\sigma
	(Z)F_{\epsilon ^{\ast }}^{-1}(V),  
    \end{aligned}
\end{eqnarray}%
where $Y$ denotes the observed budget share for alcohol, fares, and fuel,
respectively. The key endogenous regressor, $R$, is the logarithm of total
expenditure. The corresponding coefficient $\alpha (\tau )$ captures the
heterogeneous income elasticities of interest, which is allowed to vary with relative ranks of commodity expenditure. To cope with endogeneity, \citet{gorman1959separable} argued that labor income is a valid instrument for total expenditure. See also \citet{blundell2007semi}
and \citet{chernozhukov2015quantile}. Following this literature, we use
the household head's earnings in 1995 as the excluded instrument $Z_{1}$ without further verification. $X$ include
the household head's age, and a demographic dummy indicating whether the
household has a child or not. This regression
specification is a special case of our model \ref{censor2}-\ref{endo2}.

We implement the TNS algorithm to estimate the distributional effects. In the first-step estimation, we assume
that $\epsilon ^{\ast }$ follows a standard normal distribution for
simplicity and set $(\tilde{v}_{1},\tilde{v}_{2})=(0.25,0.75)$. In the
second-step estimation, 1\% outliers for $(X_{i},R_{i},Z_{i})$ are removed.
The function $m(v,\tau )$ is approximated by a third-order raw polynomial
series expansion. The grid is constructed with a step size of $\tau
_{l-1}-\tau _{l}=0.02$ for $l=1,\cdots ,L$ with $\tau _{0}=0.98$ and $\tau
_{L}=0.3$. The tunning parameter $\tilde{\eta}_{n}$ is initialized as the $%
0.1$-th quantile of the positive values of $X^{\prime }\hat{\beta}(\tau
_{l-1})+R\hat{\alpha}(\tau _{l-1})+\hat{m}(\hat{V},\tau _{l-1})$ and updated
to the $0.03$-th quantile of the positive values of $X^{\prime }\tilde{\beta}%
(\tau _{l})+R\tilde{\alpha}(\tau _{l})+\tilde{m}(\hat{V},\tau _{l})$.
Standard errors are computed using the proposed weighted bootstrap method
with $T=999$ replications.

The left panel of Table \ref{table3} reports the TNS estimates and standard errors
for $\tau =(0.9,0.7,0.5,0.3)$. The results indicate that expenditure
elasticities are heterogeneous across relative ranks of commodity-specific expenditure.
For alcohol, households at intermediate ranks exhibit larger elasticities with respect to total expenditure than those at the top or bottom of the distribution. In contrast, for fares, households at intermediate ranks substantially reduce the budget share devoted to travel as total expenditure rises, whereas households at high or low ranks display comparatively modest adjustments in travel-related consumption. For fuel, all households reduce their budget shares as total expenditure increases, with the largest declines among high-rank fuel consumers. These findings have
important implications for both corporate marketing and public policy, as
they provide evidence for the design of differentiated marketing and pricing
strategies across household segments. Figure 1 depicts the estimated expenditure elasticities over quantile levels.
\begin{table}
	\centering
	\small
	\caption{Expenditure Elasticities at Quantiles}
	\label{table3}
	\begin{tabular}{ccccccccc}
		\hline
		${\small \tau }$ & ${\small 0.9}$ & ${\small 0.7}$ & ${\small 0.5}$ & $%
		{\small 0.3}$ & ${\small 0.9}$ & ${\small 0.7}$ & ${\small 0.5}$ & ${\small %
			0.3}$ \\ \cline{2-9}
		${\small Y}$ & \multicolumn{8}{c}{\small Alcohol Expenditure} \\ 
		$\log ${\small exp} & {\small -0.006} & {\small 0.015}$^{\ast \ast \ast }$ & 
		{\small 0.016}$^{\ast \ast \ast }$ & {\small 0.013}$^{\ast \ast \ast }$ & 
		{\small 0.142}$^{\ast \ast \ast }$ & {\small 0.120}$^{\ast \ast \ast }$ & 
		{\small 0.107}$^{\ast \ast \ast }$ & {\small 0.071}$^{\ast \ast \ast }$ \\ 
		& {\small (0.004)} & {\small (0.002)} & {\small (0.001)} & {\small (0.001)}
		& {\small (0.054)} & {\small (0.019)} & {\small (0.017)} & {\small (0.017)}
		\\ 
		{\small (}$\log ${\small exp)}$^{2}$ & - & - & - & - & {\small -0.014}$%
		^{\ast \ast \ast }$ & {\small -0.010}$^{\ast \ast \ast }$ & {\small -0.008}$%
		^{\ast \ast \ast }$ & {\small -0.005}$^{\ast \ast \ast }$ \\ 
		& - & - & - & - & {\small (0.005)} & {\small (0.002)} & {\small (0.002)} & 
		{\small (0.001)} \\ 
		{\small child} & {\small -0.054}$^{\ast \ast \ast }$ & {\small -0.031}$%
		^{\ast \ast \ast }$ & {\small -0.017}$^{\ast \ast \ast }$ & {\small -0.009}$%
		^{\ast \ast \ast }$ & {\small -0.052}$^{\ast \ast \ast }$ & {\small -0.029}$%
		^{\ast \ast \ast }$ & {\small -0.017}$^{\ast \ast \ast }$ & {\small -0.010}$%
		^{\ast \ast \ast }$ \\ 
		& {\small (0.004)} & {\small (0.003)} & {\small (0.002)} & {\small (0.001)}
		& {\small (0.004)} & {\small (0.003)} & {\small (0.002)} & {\small (0.001)}
		\\ 
		{\small age} & {\small -0.001}$^{\ast \ast \ast }$ & {\small -0.001}$^{\ast
			\ast \ast }$ & {\small -0.000}$^{\ast \ast \ast }$ & {\small -0.000}$^{\ast
			\ast \ast }$ & {\small -0.001}$^{\ast \ast \ast }$ & {\small -0.001}$^{\ast
			\ast \ast }$ & {\small -0.000}$^{\ast \ast \ast }$ & {\small -0.000}$^{\ast
			\ast \ast }$ \\ 
		& {\small (0.000)} & {\small (0.000)} & {\small (0.000)} & {\small (0.000)}
		& {\small (0.000)} & {\small (0.000)} & {\small (0.000)} & {\small (0.000)}
		\\ \cline{2-9}
		${\small Y}$ & \multicolumn{8}{c}{\small Fares Expenditure} \\ 
		$\log ${\small exp} & {\small -0.004} & {\small -0.010}$^{\ast \ast \ast }$
		& {\small -0.002}$^{\ast \ast \ast }$ & {\small 0.001}$^{\ast \ast }$ & 
		{\small 0.002} & {\small -0.043}$^{\ast \ast \ast }$ & {\small -0.025}$%
		^{\ast \ast \ast }$ & {\small -0.012}$^{\ast \ast \ast }$ \\ 
		& {\small (0.004)} & {\small (0.001)} & {\small (0.000)} & {\small (0.000)}
		& {\small (0.036)} & {\small (0.012)} & {\small (0.007)} & {\small (0.004)}
		\\ 
		{\small (}$\log ${\small exp)}$^{2}$ & - & - & - & - & {\small -0.001} & 
		{\small 0.003}$^{\ast \ast \ast }$ & {\small 0.002}$^{\ast \ast \ast }$ & 
		{\small 0.001}$^{\ast \ast \ast }$ \\ 
		& - & - & - & - & {\small (0.003)} & {\small (0.001)} & {\small (0.001)} & 
		{\small (0.000)} \\ 
		{\small child} & {\small -0.011}$^{\ast \ast }$ & {\small 0.003}$^{\ast }$ & 
		{\small 0.003}$^{\ast \ast \ast }$ & {\small 0.004}$^{\ast \ast }$ & {\small %
			-0.011}$^{\ast \ast \ast }$ & {\small 0.002} & {\small 0.004}$^{\ast \ast
			\ast }$ & {\small 0.004}$^{\ast \ast \ast }$ \\ 
		& {\small (0.004)} & {\small (0.002)} & {\small (0.001)} & {\small (0.002)}
		& {\small (0.004)} & {\small (0.002)} & {\small (0.001)} & {\small (0.002)}
		\\ 
		{\small age} & {\small -0.001}$^{\ast \ast \ast }$ & {\small -0.000}$^{\ast
			\ast \ast }$ & {\small -0.000}$^{\ast \ast \ast }$ & {\small 0.000}$^{\ast
			\ast }$ & {\small -0.001}$^{\ast \ast \ast }$ & {\small -0.000}$^{\ast \ast
			\ast }$ & {\small -0.000}$^{\ast \ast \ast }$ & {\small 0.000}$^{\ast \ast }$
		\\ 
		& {\small (0.000)} & {\small (0.000)} & {\small (0.000)} & {\small (0.000)}
		& {\small (0.000)} & {\small (0.000)} & {\small (0.000)} & {\small (0.000)}
		\\ \cline{2-9}
		${\small Y}$ & \multicolumn{8}{c}{\small Fuel Expenditure} \\ 
		$\log ${\small exp} & {\small -0.084}$^{\ast \ast \ast }$ & {\small -0.055}$%
		^{\ast \ast \ast }$ & {\small -0.039}$^{\ast \ast \ast }$ & {\small -0.027}$%
		^{\ast \ast \ast }$ & {\small -0.393}$^{\ast \ast \ast }$ & {\small -0.238}$%
		^{\ast \ast \ast }$ & {\small -0.156}$^{\ast \ast \ast }$ & {\small -0.094}$%
		^{\ast \ast \ast }$ \\ 
		& {\small (0.002)} & {\small (0.001)} & {\small (0.001)} & {\small (0.001)}
		& {\small (0.028)} & {\small (0.017)} & {\small (0.014)} & {\small (0.011)}
		\\ 
		{\small (}$\log ${\small exp)}$^{2}$ & - & - & - & - & {\small 0.029}$^{\ast
			\ast \ast }$ & {\small 0.017}$^{\ast \ast \ast }$ & {\small 0.011}$^{\ast
			\ast \ast }$ & {\small 0.006}$^{\ast \ast \ast }$ \\ 
		& - & - & - & - & {\small (0.003)} & {\small (0.002)} & {\small (0.001)} & 
		{\small (0.001)} \\ 
		{\small child} & {\small 0.019}$^{\ast \ast \ast }$ & {\small 0.013}$^{\ast
			\ast \ast }$ & {\small 0.013}$^{\ast \ast \ast }$ & {\small 0.012}$^{\ast
			\ast \ast }$ & {\small 0.014}$^{\ast \ast \ast }$ & {\small 0.011}$^{\ast
			\ast \ast }$ & {\small 0.011}$^{\ast \ast \ast }$ & {\small 0.010}$^{\ast
			\ast \ast }$ \\ 
		& {\small (0.002)} & {\small (0.001)} & {\small (0.001)} & {\small (0.001)}
		& {\small (0.002)} & {\small (0.001)} & {\small (0.001)} & {\small (0.001)}
		\\ 
		{\small age} & {\small 0.000}$^{\ast \ast \ast }$ & {\small 0.000}$^{\ast
			\ast \ast }$ & {\small 0.000}$^{\ast \ast \ast }$ & {\small 0.000}$^{\ast
			\ast \ast }$ & {\small 0.000}$^{\ast \ast \ast }$ & {\small 0.000}$^{\ast
			\ast \ast }$ & {\small 0.000}$^{\ast \ast \ast }$ & {\small 0.000}$^{\ast
			\ast \ast }$ \\ 
		& {\small (0.000)} & {\small (0.000)} & {\small (0.000)} & {\small (0.000)}
		& {\small (0.000)} & {\small (0.000)} & {\small (0.000)} & {\small (0.000)}
		\\ 
		{\small Obs.} & {\small 6797} & {\small 6797} & {\small 6797} & {\small 6797}
		& {\small 6797} & {\small 6797} & {\small 6797} & {\small 6797} \\ \hline
	\end{tabular}%
\end{table}

\begin{figure}
	\centering
	\includegraphics[width=130mm,height=60mm]{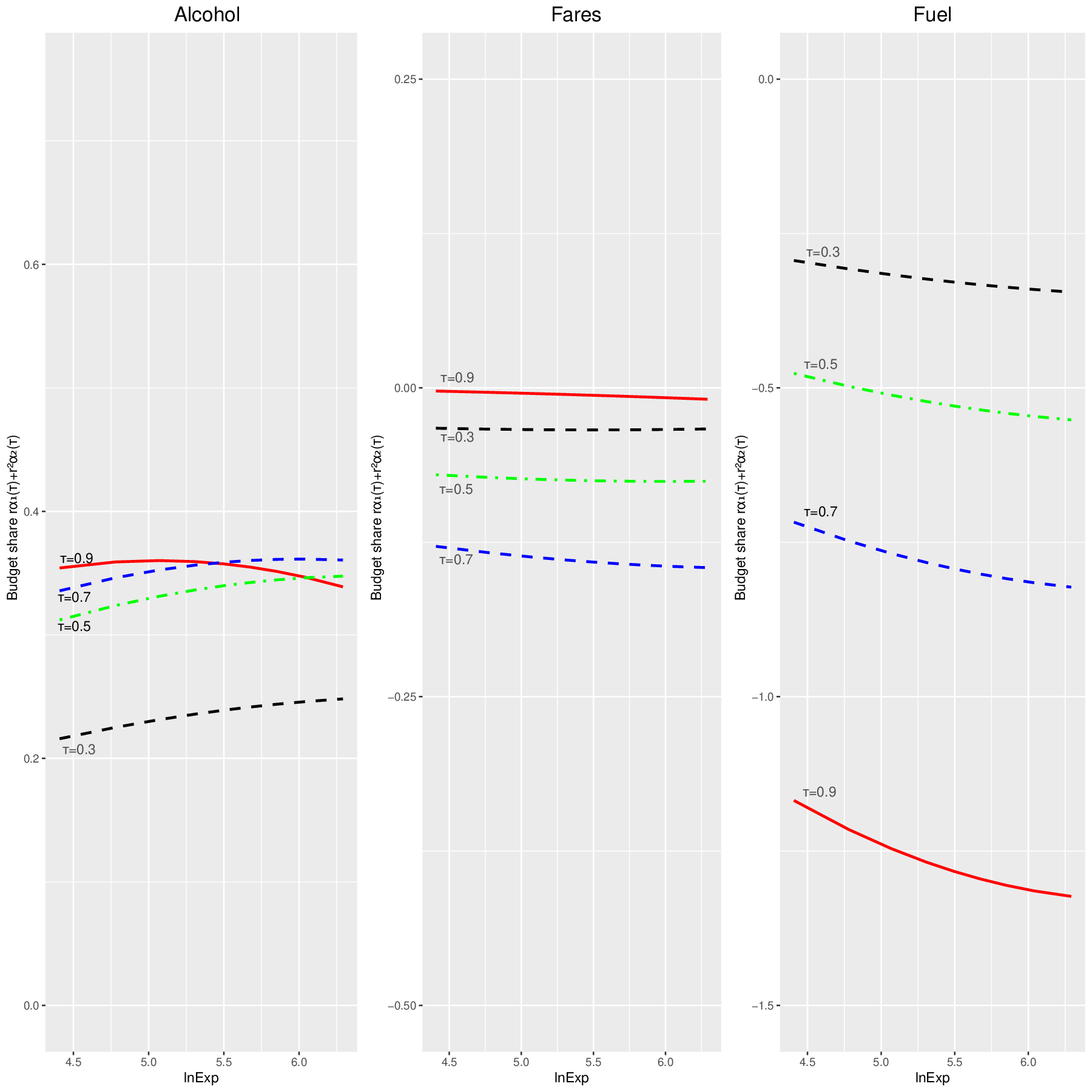}
	\caption{Expenditure Pattern}
	\label{fig2}
\end{figure}
 Some Engel curve studies, such as \citet{hausman1995nonlinear} and \citet{blundell2007semi},
 found that average expenditure elasticities vary with the level of
(ln)total expenditure. They attribute the heterogeneity in
spending behaviors to the nonlinear influence of total expenditure. A
natural question arises as to whether the elasticities exhibit similar
patterns when household heterogeneity is explicitly captured
through quantile regression. To examine this viewpoint, we augment the
baseline specification by incorporating a quadratic term of $R$:%
\begin{eqnarray}
	\begin{aligned}\label{spec2}
		Y &=\max \{X^{\prime }\beta (U)+R\alpha _{1}(U)+R^{2}\alpha
		_{2}(U)+m(V,U),0\}, \\
		R &=Z^{\prime }\gamma +\sigma (Z)F_{\epsilon ^{\ast }}^{-1}(V), 
	\end{aligned}
\end{eqnarray}%
The extended specification enhances flexibility in capturing distributional
effects, allowing the elasticity to vary not only across the conditional
distribution of each commodity expenditure but also with the logarithm of
total expenditure.

\begin{figure}
	\centering
	\subfloat{\includegraphics[width=65mm]{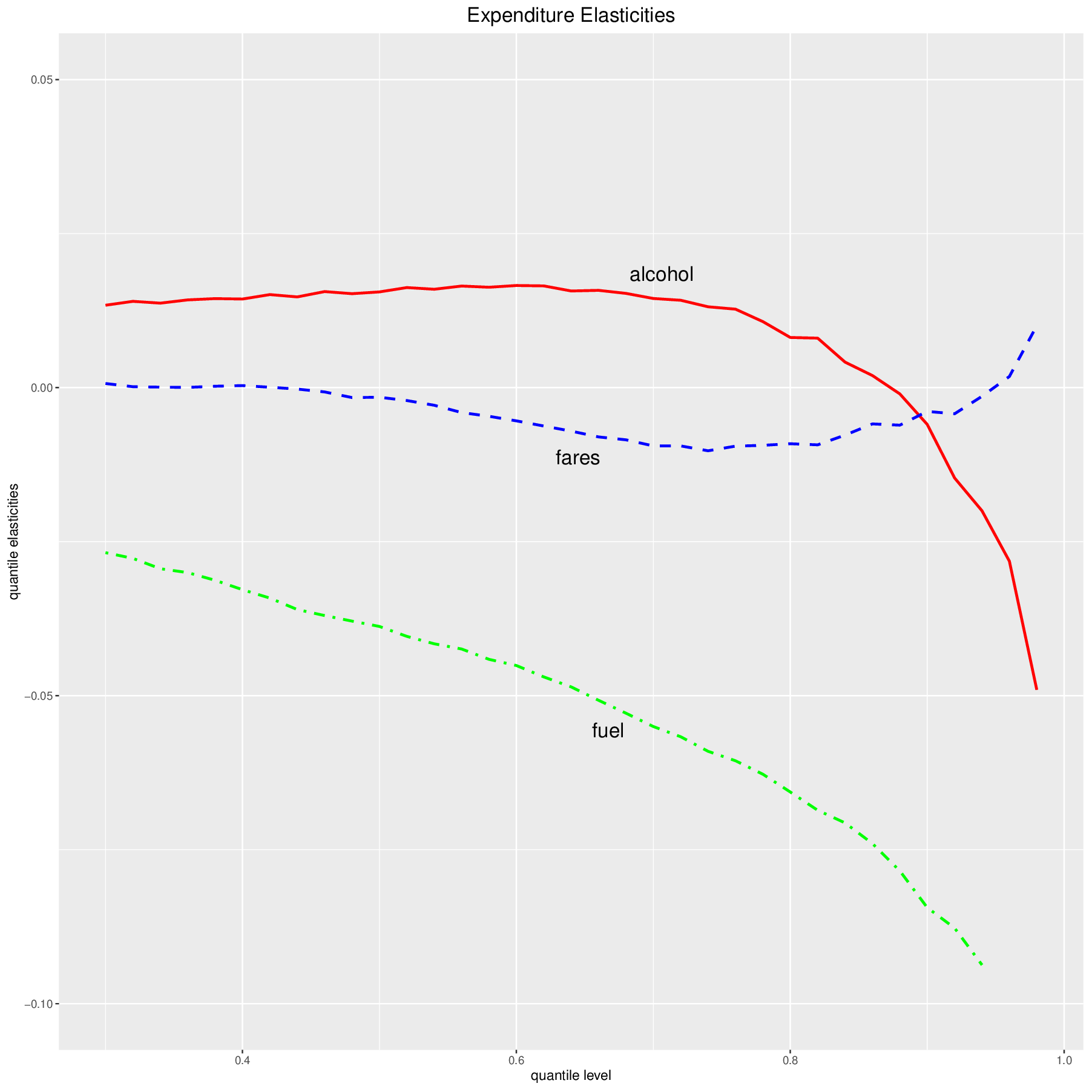}}
	\subfloat{\includegraphics[width=65mm]{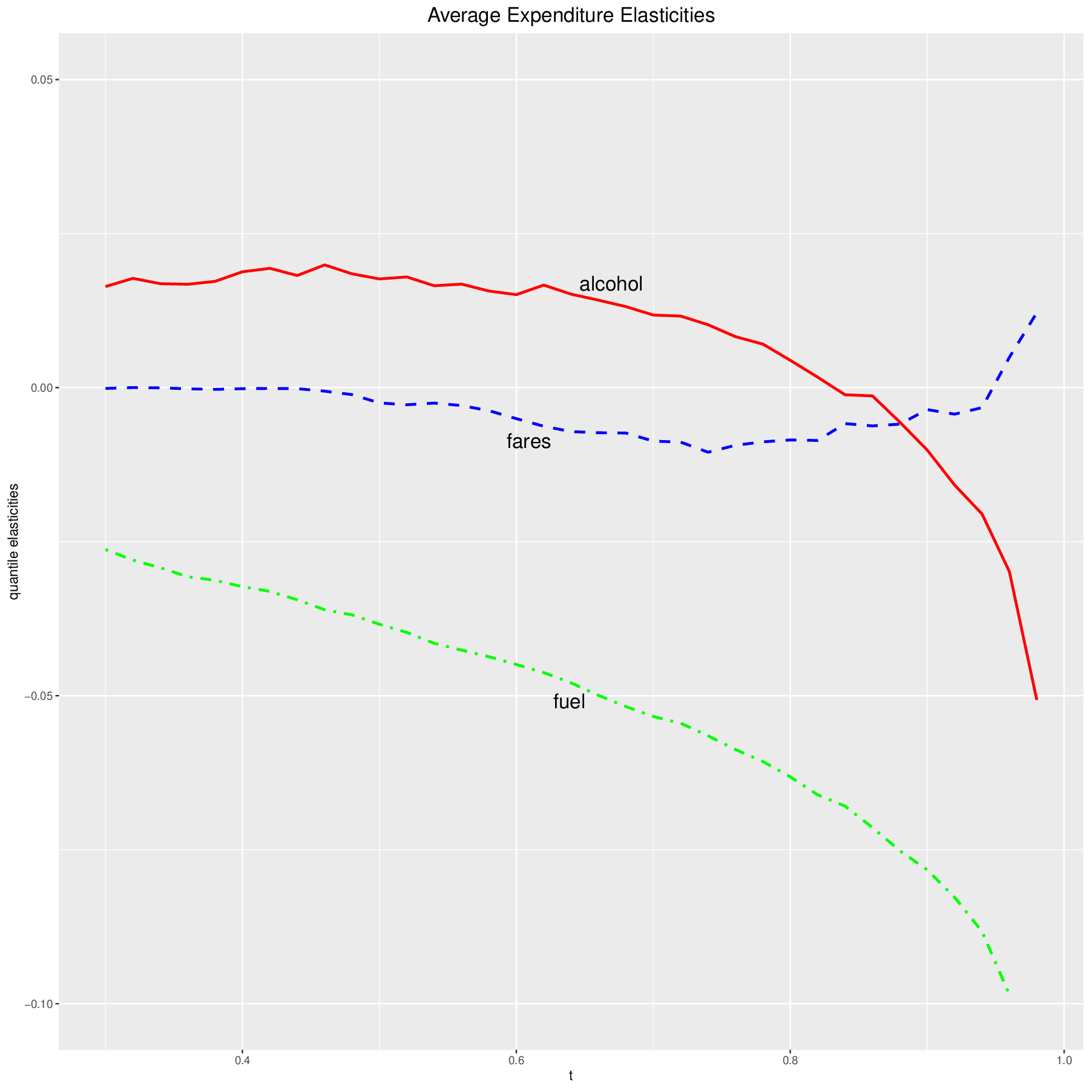}}
	\caption{Expenditure Elasticities}
	\label{fig1}
\end{figure}
The right panel of Table \ref{table3} presents the TNS estimates and their standard errors. For each
quantile level, the relationship between the budget share of alcohol and the
total expenditure exhibits an inverted-U pattern. That means households tend
to increase, then subsequently decrease, the budget share devoted to alcohol
as total expenditure rises. By contrast, the budget shares of fares and fuel
display a U-shaped pattern in relation to total expenditure---households
initially reduce, and later increase, the shares allocated to fares and fuel
as total expenditure grows.

Despite some nonlinear effects, $\left\vert \alpha _{2}(\tau )\right\vert $
are substantially smaller than $\left\vert \alpha _{1}(\tau )\right\vert $. The
right-hand panel of Figure 1 displays the average expenditure elasticities
across quantile levels, calculated as $\alpha _{1}(\tau )+2\alpha _{2}(\tau )%
\frac{1}{n}\sum_{i=1}^{n}R_{i}$. The average distributional effects closely
resemble those shown in the previous figure, suggesting that specification (%
\ref{spec1}) effectively captures the heterogeneous expenditure elasticities
of demand across households' relative ranks, averaged over logarithm of total
expenditure. Figure 2 illustrates the budget share explained by logarithm of
total expenditure for each quantile level, expressed as $r\alpha _{1}(\tau
)+r^{2}\alpha _{2}(\tau )$. The distinct curves confirm that households with
different relative ranks display different expenditure patterns. Further, the
near-parallel alignment of most curves indicates that specification (\ref%
{spec1}) is suitable to capture the quantile income elasticities of demand.

\section{Conclusion}

Distributional effects provide valuable insights into how economic factors
impact individuals differently across various segments in the outcome
distribution. Such effects are particularly important\ to reveal
heterogeneity in applied research. At the same time, censoring, endogeneity and heteroskedasticity are common in empirical data
analysis. Estimating the distributional effects in a censored quantile model
with endogeneity and heteroskedasticity is therefore an appealing and technically demanding task,
particularly when the endogenous regressor and the latent outcome have
conditional heteroskedastic errors. This paper employed the conditional CDF as
the control variable and combined it with the sequential analysis. The
proposed TNS method consistently and robustly estimate distributional effects. Moreover, the algorithm is computationally
straightforward, since it merely relies on a sequence of quantile series
regressions within appropriately selected subsamples. The robustness and
tractability of the TNS estimator make it attractive to a wide range of
applied researchers in economics. Based on the derived uniform
representation, this paper also proposes a weighted bootstrap procedure for
TNS estimation.

Panel data analysis has received a great deal of attention in recent years.
\citet{gu2019panel} developed an estimator for panel data quantile
models with group fixed effects. However, their method fails in the presence
of censored data. The estimation of censored panel data models incorporating
endogeneity remains an important and challenging topic, which we leave for
future research.

\bibliographystyle{chicago}
\bibliography{ref}

\appendix
\begin{center}
{\huge \bf Appendix}
\end{center}

\setcounter{table}{0} \renewcommand{\thetable}{A.\arabic{table}} 
\setcounter{section}{0} \renewcommand{\thesection}{A.\arabic{section}} 
\setcounter{figure}{0} \renewcommand{\thefigure}{A.\arabic{figure}}
\setcounter{remark}{0} \renewcommand{\theremark}{A.\arabic{remark}}
\setcounter{lemma}{0}
\renewcommand{\thelemma}{A.\arabic{lemma}}
\setcounter{equation}{0}
\renewcommand{\theequation}{A.\arabic{equation}}
\onehalfspacing

\section{Proofs}
Throughout this appendix, $\mathcal{S}^{\dim (\theta )-1}$ signifies the
unit sphere in $\mathbb{R}^{\dim (\theta )}$, i.e., $\mathcal{S}^{\dim
	(\theta )-1}=\{\mathbf{\iota }\in \mathbb{R}^{\dim (\theta )}:\left\Vert 
\mathbf{\iota }\right\Vert =1\}$, where $\dim (\theta )$ denotes the
dimension of $\theta $. Denote $\pi _{n}=\sqrt{\frac{J\ln n}{n}}$.
\begin{proof}[Proof of Theorem \ref{Thm1}]
As $V$ given $Z=z$ is continuously distributed and the conditional CDF is strictly increasing by Assumption ID1(i), $V$ given $Z=z$ can be normalized to have a standard uniform distribution. Then, we have
\begin{eqnarray*}
  v =P(V\leq v|Z=z)=P(m_{R}(Z,V)\leq m_{R}(z,v)|Z=z)=F_{R|Z}(r|z),
\end{eqnarray*}%
where the second equality holds by the strict monotonicity of $m_{R}(z,v)$ in $v$. Up to now, $V$ is identified as $F_{R|Z}(R|Z)$. 

On the one hand, 
\begin{eqnarray*}
	Q_{Y|R,Z,D_{\tau }(\eta )=1}(\tau |r,z(v,x,r))&=&Q_{Y^{\ast }|R,Z}(\tau |r,z(v,x,r))\\
                                           &=&Q_{Y^{\ast }|X,R,V}(\tau |x,r,v) \\
                                           &=&m_{Y^{\ast }}(x,r,v,\tau ),
\end{eqnarray*}%
where the first equality holds since $\tau$th quantile of $Y$ coincides with that of $Y^{\ast}$ when $D_{\tau }(\eta )=1$; the second equality and final equality hold
by Assumption ID(ii).

On the other hand, as $V=F_{R|Z}(R|Z)$, for any $V$, there exists
$Z=z(V,R)$ such that $V=F_{R|Z}(R|z(V,R))$. Under Assumption ID1(iii)
\begin{equation*}
	\nabla _{r}z(v,r)=\nabla _{z}Q_{R|Z}(v|z)[\nabla _{z}Q_{R|Z}(v|z)^{\prime
	}\nabla _{z}Q_{R|Z}(v|z)]^{-1}.
\end{equation*}%
By the chain rule, we obtain%
\begin{eqnarray*}
	\nabla _{r}Q_{Y^{\ast }|X,R,V}(\tau |x,r,v) &=&\nabla _{r}Q_{Y|R,Z,D_{\tau }(\eta
		)=1}(\tau |r,z)+\nabla _{z}Q_{R|Z}(v|z)[\nabla _{z}Q_{R|Z}(v|z)^{\prime
	}\nabla _{z}Q_{R|Z}(v|z)]^{-1} \\
	&&\times \nabla _{z}Q_{Y|R,Z,D_{\tau }(\eta )=1}(\tau |r,z).
\end{eqnarray*}%
\end{proof}

\begin{proof}[Proof of Theorem \ref{Thm2}]
	
	The monotonicity condition in Assumption ID2(i) and exclusion
	condition in Assumption ID2(ii) guarantee the existence of $%
	Q_{R|Z}(v|z)=h(z)+\sigma(z)Q_{\varepsilon^{\ast}}(v)$ and $Q_{Y^{\ast }|X,R,V}(\tau |x,r,v)=x^{\prime }\beta (\tau
	)+r\alpha (\tau )+m(v,\tau )$. When $\tau \geq \tau _{0}$, the initial
	condition in Assumption ID2(iv) implies $D_{\tau }(\eta )=1$ $a.s.$%
	, then it is straightforward to identify $(\beta (\tau ),\alpha (\tau ))$ under the
	full rank condition in Assumption ID2(iii).
	
	Next, consider the quantile level $\tau _{1}<\tau _{0}$ where the censoring
	comes into play. If $\tilde{\eta}>\left\vert \eta-\sup_{x,r,v}\left\vert Q_{Y^{\ast }|X,R,V}
(\tau _{1}|x,r,v)-Q_{Y^{\ast}|X,R,V}(\tau _{0}|x,r,v)\right\vert \right\vert $, then
    $D_{\tau _{1}}(\eta )\geq
	D_{\tau _{1}}(\tau _{0},\tilde{\eta})$ $a.s.$ for any $\eta \geq 0$.The continuity condition in Assumption ID2(v) guarantees existence of $D_{\tau _{1}}(\tau _{0},\tilde{\eta})=1$. Given the fact that $D_{\tau _{1}}(\tau _{0},\tilde{\eta})=1$ implies $D_{\tau_{1}}(\eta )=1$, we have
	\begin{equation*}
		Q_{Y|X,R,V,D_{\tau _{1}}(\tau _{0},\tilde{\eta})=1}(\tau
		_{1}|X,R,V)=X^{\prime }\beta (\tau _{1})+R\alpha (\tau _{1})+m(V,\tau _{1}).
	\end{equation*}%
	since 	
	\begin{eqnarray*}
		&&P(Y\leq X^{\prime }\beta (\tau _{1})+R\alpha (\tau _{1})+m(V,\tau _{1})|X,R,V,Q_{Y^{\ast }|X,R,V}(\tau_{0}|X,R,V)>\tilde{\eta})\\
		&=&P(Y\leq X^{\prime }\beta (\tau _{1})+R\alpha (\tau _{1})+m(V,\tau _{1})|X,R,V,Q_{Y^{\ast }|X,R,V}(\tau_{1}|X,R,V)> \eta)\\
		&=&\tau_{1}.
	\end{eqnarray*}%
	If there exists $(b,a,\tilde{m}(v,\tau ))$ such that 
	\begin{equation*}
		Q_{Y|X,R,V,D_{\tau _{1}}(\tau _{0},\tilde{\eta})=1}(\tau |X,R,V)=X^{\prime
		}b+Ra+\tilde{m}(V,\tau _{1}),
	\end{equation*}%
	then 
	\begin{equation}
		X^{\prime }b+Ra+\tilde{m}(V,\tau _{1})=X^{\prime }\beta (\tau _{1})+R\alpha
		(\tau _{1})+m(V,\tau _{1}).  \tag{A1}\label{A1} 
	\end{equation}%
	Define an operator $A^{\ast }(\tau _{1})=\frac{E\left[ \left. D_{\tau
			_{1}}(\tau _{0},\tilde{\eta})f_{Y^{\ast }|X,R,V}(Q_{Y^{\ast }|X,R,V}(\tau
		_{1}|X,R,V))A\right\vert V\right] }{E\left[ \left. D_{\tau _{1}}(\tau _{0},%
		\tilde{\eta})f_{Y^{\ast }|X,R,V}(Q_{Y^{\ast }|X,R,V}(\tau
		_{1}|X,R,V))\right\vert V\right] }$ for any variable $A$. Apply the operator
	to (\ref{A1}), then 
	\begin{equation}
		X^{\ast }(\tau _{1})^{\prime }b+R^{\ast }(\tau _{1})a+\tilde{m}(V,\tau
		_{1})=X^{\ast }(\tau _{1})^{\prime }\beta (\tau _{1})+R^{\ast }(\tau
		_{1})\alpha (\tau _{1})+m(V,\tau _{1}).  \tag{A2}\label{A2}
	\end{equation}%
	Substracting (\ref{A2}) from (\ref{A1}) yields 
	\begin{equation}
		(X-X^{\ast }(\tau _{1}))^{\prime }b+(R-R^{\ast }(\tau _{1}))a=(X-X^{\ast
		}(\tau _{1}))^{\prime }\beta (\tau _{1})+(R-R^{\ast }(\tau _{1}))\alpha
		(\tau _{1}).  \tag{A3}\label{A3}
	\end{equation}%
	By Assumption ID2(iii) and (v), $E[D_{\tau _{1}}(\eta )f_{Y^{\ast
		}|X,R,V}(Q_{Y^{\ast }|X,R,V}(\tau _{1}|X,R,V))P(\tau _{1})P(\tau
	_{1})^{\prime }]$ is positive definite. Subsequently, we shall identify $%
	(\beta (\tau _{1}),\alpha (\tau _{1}))$.
	
	Following the same line, $(\beta (\tau _{2}),\alpha (\tau _{2})),\cdots
	,(\beta (\tau _{L}),\alpha (\tau _{L}))$ are identified in sequence. For any 
	$\tau \in \lbrack \tau _{L},\tau _{0})$, there exists $l$ such that $\tau
	\in \lbrack \tau _{l},\tau _{l-1})$. Then, we identify $(\beta (\tau
	),\alpha (\tau ))$ using the above argument by replacing $\tau _{0}$ and $%
	\tau _{1}$ with $\tau _{l-1}$ and $\tau $.
\end{proof}

\begin{proof}[Proof of Theorem \ref{Thm3}]
	
	\textit{Step 1. Uniform consistency of }$\hat{\theta}(\tau )$ \textit{for} $%
	\tau \geq \tau _{0}$.
	
	With slight abuse of notations, let $S_{n}(\theta ,\tau )=\frac{1}{n}%
	\sum_{i=1}^{n}T_{i}\rho _{\tau }(Y_{i}-\hat{W}_{i}^{\prime }\theta )$ with $%
	\hat{W}_{i}=(X_{i}^{\prime },R_{i},P_{J}(\hat{V}_{i}{\small )}^{\prime
	})^{\prime }$ and $\theta =(b^{\prime },a,d^{\prime })^{\prime }$. Decompose 
	\begin{equation*}
		S_{n}(\theta ,\tau )-S_{n}(\theta (\tau ),\tau )=S(\theta ,\tau )-S(\theta
		(\tau ),\tau )+\bar{S}_{n}(\theta ,\theta (\tau ),\tau )
	\end{equation*}%
	where $S(\theta ,\tau )=E[S_{n}(\theta ,\tau )]$ and $\bar{S}_{n}(\theta
	,\theta (\tau ),\tau )=S_{n}(\theta ,\tau )-S_{n}(\theta (\tau ),\tau
	)-S(\theta ,\tau )+S(\theta (\tau ),\tau )$.
	
	On the one hand, define a class of functions as 
	\begin{equation*}
		\mathcal{F}_{\bar{s}}=\{f_{\bar{s}}:(T,W,Y)\longmapsto T\rho _{\tau
		}(Y-W^{\prime }\theta )-T\rho _{\tau }(Y-W^{\prime }\theta (\tau )),\theta
		\in \mathcal{\Theta },\tau \in \lbrack \tau _{0},\tau _{U}]\}
	\end{equation*}%
	where $\mathcal{\Theta =B\times A\times D}$. Under Assumption A1 and
	Assumption A2, there exists a function $M(Y,W)$ such that%
	\begin{equation*}
		\left\vert T\rho _{\tau }(Y-W^{\prime }\theta )-T\rho _{\tau }(Y-\hat{W}%
		^{\prime }\theta (\tau ))\right\vert \leq TM(Y,W)\left\Vert \theta -\theta
		(\tau )\right\Vert ,
	\end{equation*}%
	where the envelop of $\mathcal{F}_{\bar{s}}$ satisfies $\bar{F}_{\bar{s}%
	}\lesssim \zeta _{0}$ and $E\left[ TM(Y,W)^{2}\right] \lesssim J$. By Lemma
	26 of \citet{belloni2019conditional}, the covering number of $\mathcal{F}_{\bar{s}}$
	is 
	\begin{equation*}
		\mathcal{N(}\left\Vert f_{\bar{s}}\right\Vert _{2},L_{2}(P),\bar{F}_{\bar{s}%
		}\lambda \mathcal{)\lesssim }(c_{\mathcal{N}}/\lambda )^{J}
	\end{equation*}%
	for a constant $c_{\mathcal{N}}$ and $0<\lambda <1$. Taking $x=3\ln n$ of
	Lemma S1, we have 
	\begin{equation*}
		P\left( \sup_{\theta \in \Theta }\left\vert \bar{S}_{n}(\theta ,\theta (\tau
		),\tau )\right\vert >c_{\bar{s}}\pi _{n}\right) \leq \frac{1}{n^{3}}
	\end{equation*}%
	for a constant $c_{\bar{s}}>0$.
	
	On the other hand, let $H_{n}(\theta (\tau ),\tau )=\frac{1}{n}%
	\sum_{i=1}^{n}T_{i}\varphi _{\tau }(Y_{i}-\hat{W}_{i}^{\prime }\theta (\tau
	))\hat{W}_{i}$ and $H(\theta (\tau ),\tau )=E\left[ H_{n}(\theta (\tau
	),\tau )\right] $, where $\varphi _{\tau }(u)=1\{u<0\}-\tau $. We have
	$E[T(F_{Y|X,R,V}(\hat{W}^{\prime }\theta (\tau ))-\tau )\hat{W}%
	]\lesssim J^{1/2}(J^{-\kappa }+n^{-1/2})$\footnote{%
		If $\hat{v}(r,x)$ is a parametric estimator, then $\pi _{1n}=n^{-1/2}$; if $%
		\hat{v}(r,x)$ is an undersmoothing nonparametric estimator, then $%
		E[T(F_{Y|X,R,V}(\hat{W}^{\prime }\theta (\tau ))-\tau )\hat{W}]\lesssim
		J^{1/2}(J^{-\kappa }+\pi _{1n}^{2})$} under Assumption A4-A6. Then,
	\begin{eqnarray*}
		&&S(\hat{\theta}(\tau ),\tau )-S(\theta (\tau ),\tau ) \\
		&=&H(\theta (\tau ),\tau )\left( \hat{\theta}(\tau )-\theta (\tau )\right)
		+\left( \hat{\theta}(\tau )-\theta (\tau )\right) ^{\prime }E\left[
		Tf_{Y^{\ast }|X.R,V}(\hat{W}^{\prime }\bar{\theta})\hat{W}\hat{W}^{\prime }%
		\right] \left( \hat{\theta}(\tau )-\theta (\tau )\right) \\
		&\geq &-c_{\underline{s}}\left( J^{1/2-\kappa }+n^{-1/2}\right) \left\Vert 
		\hat{\theta}(\tau )-\theta (\tau )\right\Vert +\frac{1}{2}\underline{c}%
		_{\Omega }\left\Vert \hat{\theta}(\tau )-\theta (\tau )\right\Vert ^{2}
	\end{eqnarray*}%
	uniformly in $\tau \in \lbrack \tau _{L},\tau _{U}]$, where the final line
	holds by $E\left[ Tf_{Y^{\ast }|X,R,V}(\hat{W}^{\prime }\theta (\tau ))%
	\hat{W}\hat{W}^{\prime }\right] \\
	>\frac{1}{2}\underline{c}_{\Omega }$ and
	the fact that $\left\Vert \hat{W}-W\right\Vert \overset{p}{%
		\rightarrow }0$.
	
	As $\hat{\theta}(\tau )$ minimizes $S_{n}(\theta ,\tau )$ over $\theta \in 
	\mathcal{\Theta }$, we have for all $\tau $%
	\begin{eqnarray*}
		0 &\geq &S_{n}(\hat{\theta}(\tau ),\tau )-S_{n}(\theta (\tau ),\tau ) \\
		&\geq &-c_{\bar{s}}\pi _{n}-c_{\underline{s}}J^{1/2}(J^{-\kappa
		}+n^{-1/2})\left\Vert \hat{\theta}(\tau )-\theta (\tau )\right\Vert +\frac{1%
		}{2}\underline{c}_{\Omega }\left\Vert \hat{\theta}(\tau )-\theta (\tau
		)\right\Vert ^{2}
	\end{eqnarray*}%
	with probability no less than $1-\frac{1}{n^{3}}$. Subsequently, $\left\Vert 
	\hat{\theta}(\tau )-\theta (\tau )\right\Vert =O_{p}(\tilde{\pi}_{n})$ where 
	$\tilde{\pi}_{n}=\left( \frac{J\ln n}{n}\right) ^{1/4}+J^{1/2}(J^{-\kappa
	}+n^{-1/2})$. Thus, the uniform consistency is obtained under Assumption
	A7.\medskip
	
	\textit{Step 2. The uniform convergence rate of }$\hat{\theta}(\tau )$%
	\textit{\ for }$\tau >\tau _{0}$.
	
	Let 
	\begin{eqnarray*}
		R_{n}(\theta ,\theta (\tau ),\tau ) &=&\frac{1}{\left\Vert \theta -\theta
			(\tau )\right\Vert }\left[ S_{n}(\theta ,\tau )-S_{n}(\theta (\tau ),\tau
		)-H_{n}(\theta (\tau ),\tau )^{\prime }(\theta -\theta (\tau ))\right] \\
		&\equiv &\frac{1}{n}\sum_{i=1}^{n}r_{i}(\theta ,\theta (\tau ),\tau ).
	\end{eqnarray*}%
	Define a class of functions as 
	\begin{eqnarray*}
		\mathcal{F}_{r}=\{f_{r}:(T,W,Y)\longmapsto r(\theta ,\theta (\tau ),\tau
		),\left\Vert \theta -\theta (\tau )\right\Vert \lesssim \tilde{\pi}_{n},\tau
		\in \lbrack \tau _{0},\tau _{U}]\}.
	\end{eqnarray*}%
	The envelope of $\mathcal{F}_{r}$ is bounded by $\zeta _{0}$ and $E\left(
	f_{r}^{2}\right) \lesssim J$. Taking $x=3\ln n$ in Lemma S1
	together with $E[f_{r}^{2}(T,\hat{W},Y)]\lesssim J$ yields 
	\begin{eqnarray*}
		\sup_{\tau \in \lbrack \tau _{L},\tau _{U}]}\sup_{\left\Vert \theta -\theta
			(\tau )\right\Vert \lesssim \tilde{\pi}_{n}}\left\vert R_{n}(\theta ,\theta
		(\tau ),\tau )-E(R_{n}(\theta ,\theta (\tau ),\tau ))\right\vert \lesssim
		\pi _{n}
	\end{eqnarray*}%
	with the probability no less than $1-\frac{1}{n^{3}}$. As $E\left[
	\left\Vert T(1\{Y-\hat{W}^{\prime }\theta (\tau )<0\}-\tau )\hat{W}%
	\right\Vert ^{2}\right] \\
	\lesssim J$, the Bernstein's inequality implies that 
	\begin{equation*}
		\left\Vert H_{n}(\theta (\tau ),\tau )-H(\theta (\tau ),\tau )\right\Vert
		\lesssim \pi _{n}
	\end{equation*}%
	with the probability no less than $1-\frac{1}{n^{3}}$.
	
	As $\hat{\theta}(\tau )$ minimizes $S_{n}(\theta ,\tau )$, then%
	\begin{eqnarray*}
		0 &\geq &S_{n}(\hat{\theta}(\tau ),\tau )-S_{n}(\theta (\tau ),\tau ) \\
		&=&\left[ R_{n}(\hat{\theta}(\tau ),\theta (\tau ),\tau )-E(R_{n}(\hat{\theta%
		}(\tau ),\theta (\tau ),\tau ))\right] \left\Vert \hat{\theta}(\tau )-\theta
		(\tau )\right\Vert \\
		&&+\left[ H_{n}(\theta (\tau ),\tau )-H(\theta (\tau ),\tau )\right]
		^{\prime }\left( \hat{\theta}(\tau )-\theta (\tau )\right) +S(\hat{\theta}%
		(\tau ),\tau )-S(\theta (\tau ),\tau ) \\
		&\geq &-\left[\left\vert R_{n}(\hat{\theta}(\tau ),\theta (\tau ),\tau )-E(R_{n}(%
		\hat{\theta}(\tau ),\theta (\tau ),\tau ))\right\vert +\left\Vert
		H_{n}(\theta (\tau ),\tau )-H(\theta (\tau ),\tau )\right\Vert \right. \\
		&&\left. +c_{\underline{s}}J^{1/2}\left( J^{-\kappa }+n^{-1/2}\right) \right]
		\left\Vert \hat{\theta}(\tau )-\theta (\tau )\right\Vert +\frac{1}{2}%
		\underline{c}_{\Omega }\left\Vert \hat{\theta}(\tau )-\theta (\tau
		)\right\Vert ^{2}.
	\end{eqnarray*}%
	Consequently, there exists a constant $c_{\theta }>0$,%
	\begin{eqnarray*}
		&&P\left( \sup_{\tau \in \lbrack \tau _{0},\tau _{U}]}\left\Vert \hat{\theta}%
		(\tau )-\theta (\tau )\right\Vert >c_{\theta }\pi _{n}\right) \\
		&\leq &P\left( \sup_{\tau \in \lbrack \tau _{0},\tau _{U}]}\sup_{\left\Vert
			\theta -\theta (\tau )\right\Vert \lesssim \tilde{\pi}_{n}}\left\vert
		R_{n}(\theta ,\theta (\tau ),\tau )-E(R_{n}(\theta ,\theta (\tau ),\tau
		))\right\vert >\frac{c_{\theta }}{2}\pi _{n}\right) \\
		&&\left. +P\left( \sup_{\tau \in \lbrack \tau _{0},\tau _{U}]}\left\Vert
		H_{n}(\theta (\tau ),\tau )-H(\theta (\tau ),\tau )\right\Vert >\frac{%
			c_{\theta }}{2}\pi _{n}\right) \leq \frac{2}{n}\right.
	\end{eqnarray*}%
	under Assumption A7.
	
	\textit{Step 3. The uniform convergence rate of }$\hat{\theta}(\tau _{l})$%
	\textit{\ for }$l=1,\cdots ,L$.
	
	For any $\tilde{\theta}=(\tilde{b}^{\prime },\tilde{a},\tilde{d}^{\prime
	})^{\prime }$ such that $\left\Vert \tilde{\theta}-\theta (\tau
	_{l-1})\right\Vert \lesssim \pi _{n}$, we have%
	\begin{eqnarray*}
		&&\left. 1\{\hat{W}^{\prime }\tilde{\theta}>\tilde{\eta}_{n}\}\right. \\
		&=&1\left\{ X^{\prime }\beta (\tau _{l-1})+R\alpha (\tau _{l-1})+m(V,\tau
		_{l-1})>\tilde{\eta}_{n}-X^{\prime }(\tilde{b}-\beta (\tau _{l-1}))-R(\tilde{%
			a}-\alpha (\tau _{l-1}))\right. \\
		&&\left. -P_{J}(\hat{V})^{\prime }(\tilde{d}-\delta (\tau _{l-1}))-r_{J}(%
		\hat{V},\tau _{l-1})-m_{1}^{\prime }(V,\tau _{l-1})(\hat{V}-V)\right\} \\
		&\leq &1\left\{ X^{\prime }\beta (\tau _{l})+R\alpha (\tau _{l})+m(V,\tau
		_{l})>\tilde{\eta}_{n}-c_{\pi }\left( \zeta _{0}\pi _{n}+J^{-\kappa
		}+n^{-1/2}\right) -c_{L}\frac{1}{L_{n}}\right\} \\
		&\leq &1\left\{ X^{\prime }\beta (\tau _{l})+R\alpha (\tau _{l})+m(V,\tau
		_{l})>\eta \right\}
	\end{eqnarray*}%
	for constants $c_{\pi }>0$ and $c_{L}>0$, where $m_{1}^{\prime }(v,\tau )$
	is the first order derivative of $m(v,\tau )$ with respect to $v$. under
	Assumption A3. The above relationship illustrates that $\hat{D}_{\tau
		_{l}}=1 $ implies $D_{\tau _{l}}(\eta )=1$ with the probability approaching $%
	1$.
	
	Following the same line of Step 1, one can infer $\left\Vert \hat{\theta}%
	(\tau _{l})-\theta (\tau _{l})\right\Vert \lesssim \tilde{\pi}_{n}$ with the
	probability no less than $1-\frac{1}{n^{3}}$. Here $\left\Vert \hat{\theta}%
	(\tau _{l-1})-\theta (\tau _{l-1})\right\Vert \lesssim \pi _{n}$ is treated
	as given, since it has been established in the previous stage. For example,
	at $\tau _{1}$, $\left\Vert \hat{\theta}(\tau _{0})-\theta (\tau
	_{0})\right\Vert \lesssim \pi _{n}$ was proven by Step 1-2.
	
	Define two classes of functions:%
	\begin{eqnarray*}
		\mathcal{F}_{r_{l}} &=&\{f_{r_{l}}:(T,W,Y)\longmapsto \frac{1}{\left\Vert
			\theta -\theta (\tau _{l})\right\Vert }T1\{W^{\prime }\tilde{\theta}>\tilde{%
			\eta}_{n}\}[\rho _{\tau _{l}}(Y-W^{\prime }\theta )-\rho _{\tau
			_{l}}(Y-W^{\prime }\theta (\tau _{l})) \\
		&&\left. -\varphi _{\tau _{l}}(Y-W^{\prime }\theta (\tau _{l}))W^{\prime
		}(\theta -\theta (\tau _{l}))],\left\Vert \theta -\theta (\tau
		_{l})\right\Vert \lesssim \tilde{\pi}_{n},\left\Vert \tilde{\theta}-\theta
		(\tau _{l-1})\right\Vert \lesssim \pi _{n}\right. \}
	\end{eqnarray*}%
	and 
	\begin{eqnarray*}
		\mathcal{F}_{h_{l}}=\{f_{h_{l}}:(T,W,Y)\longmapsto T1\{W^{\prime }\tilde{%
			\theta}>\tilde{\eta}_{n}\}\varphi _{\tau _{l}}(Y-W^{\prime }\theta (\tau
		_{l}))W^{\prime }\mathbf{\iota },\left\Vert \tilde{\theta}-\theta (\tau
		_{l-1})\right\Vert \lesssim \tilde{\pi}_{n}\}.
	\end{eqnarray*}%
	In addition, let
	
	\begin{eqnarray*}
		H_{n}(\theta (\tau _{l}),\tilde{\theta},\tau _{l},\tilde{\eta}_{n})=\frac{1}{%
			n}\sum_{i=1}^{n}T_{i}1\{\hat{W}^{\prime }\tilde{\theta}>\tilde{\eta}%
		_{n}\}\varphi _{\tau _{l}}(Y_{i}-\hat{W}_{i}^{\prime }\theta (\tau _{l}))%
		\hat{W}_{i}
	\end{eqnarray*}%
	and 
	\begin{eqnarray*}
		R_{n}(\theta ,\theta (\tau _{l}),\tilde{\theta},\tau _{l},\tilde{\eta}_{n})
		&=& \frac{1}{\left\Vert \theta -\theta (\tau _{l})\right\Vert }[S_{n}(\theta ,%
		\tilde{\theta},\tau _{l},\tilde{\eta}_{n})-S_{n}(\theta (\tau _{l}),\tilde{%
			\theta},\tau _{l},\tilde{\eta}_{n})\\
		&&-H_{n}(\theta (\tau _{l}),\tilde{\theta},\tau _{l},\tilde{\eta}_{n})(\theta -\theta (\tau _{l}))]
	\end{eqnarray*}%
	where $S_{n}(\theta ,\tilde{\theta},\tau _{l},\tilde{\eta}_{n})$ is the
	shorthand notation of $S_{n}(b,a,d,\tilde{\theta},\tau _{l},\tilde{\eta}%
	_{n}) $.
	
	Since $E[f_{r_{l}}^{2}(T,W,Y)]\lesssim J$ and $E[f_{h_{l}}^{2}(T,W,Y)]%
	\lesssim J$, applying Lemma S1 with $x=3\ln n$ yields 
	\begin{equation*}
		\sup_{\left\Vert \tilde{\theta}-\theta (\tau _{l-1})\right\Vert \lesssim \pi
			_{n}}\sup_{\left\Vert \theta -\theta (\tau _{l})\right\Vert \lesssim \tilde{%
				\pi}_{n}}\left\vert R_{n}(\theta ,\theta (\tau _{l}),\tilde{\theta},\tau
		_{l},\tilde{\eta}_{n})-E[R_{n}(\theta ,\theta (\tau ),\tilde{\theta},\tau ,%
		\tilde{\eta}_{n})]\right\vert \lesssim \pi _{n}
	\end{equation*}%
	with the probability no less than $1-\frac{1}{n^{3}}$ and 
	\begin{equation*}
		\sup_{\left\Vert \tilde{\theta}-\theta (\tau _{l-1})\right\Vert \lesssim \pi
			_{n}}\left\Vert H_{n}(\theta (\tau _{l}),\tilde{\theta},\tau _{l},\tilde{\eta%
		}_{n})-E[H_{n}(\theta (\tau _{l}),\tilde{\theta},\tau _{l},\tilde{\eta}%
		_{n})]\right\Vert \lesssim \pi _{n}
	\end{equation*}%
	with the probability no less than $1-\frac{1}{n^{3}}$.
	
	As $\hat{\theta}(\tau _{l})$ minimizes $S_{n}(\theta ,\tilde{\theta},\tau
	_{l},\tilde{\eta}_{n})$, then 
	\begin{eqnarray*}
		0 &\geq &-[\left\vert R_{n}(\theta ,\theta (\tau _{l}),\tilde{\theta},\tau
		_{l},\tilde{\eta}_{n})-E\left( R_{n}(\theta ,\theta (\tau _{l}),\tilde{\theta%
		},\tau _{l},\tilde{\eta}_{n})\right) \right\vert \\
		&&+\left\Vert H_{n}(\theta (\tau _{l}),\tilde{\theta},\tau _{l},\tilde{\eta}%
		_{n})-E\left( H_{n}(\theta (\tau _{l}),\tilde{\theta},\tau _{l},\tilde{\eta}%
		_{n})\right) \right\Vert \\
		&&\left. +c_{\underline{s}}J^{1/2}\left( J^{-\kappa }+n^{-1/2}\right) \right]
		\times \left\Vert \hat{\theta}(\tau _{l})-\theta (\tau _{l})\right\Vert +%
		\frac{1}{2}\underline{c}_{\Omega }\left\Vert \hat{\theta}(\tau _{l})-\theta
		(\tau _{l})\right\Vert ^{2}.
	\end{eqnarray*}%
	It follows that%
	\begin{equation*}
		P\left( \left. \left\Vert \hat{\theta}(\tau _{l})-\theta (\tau
		_{l})\right\Vert >c_{\theta }\pi _{n}\right\vert \Omega _{l-1}\right) \leq 
		\frac{2}{n^{3}}
	\end{equation*}%
	where $\Omega _{l-1}$ signifies the event $\left\Vert \hat{\theta}(\tau
	_{l-1})-\theta (\tau _{l-1})\right\Vert \lesssim \pi _{n}$. Consequently,
	
	\begin{equation*}
		P\left( \max_{l=1,\cdots ,L}\left\Vert \hat{\theta}(\tau _{l})-\theta (\tau
		_{l})\right\Vert >c_{\theta }\pi _{n}\right) \leq \sum_{l=1}^{L}P\left(
		\left. \left\Vert \hat{\theta}(\tau _{l})-\theta (\tau _{l})\right\Vert
		>c_{\theta }\pi _{n}\right\vert \Omega _{l-1}\right) \leq \frac{2L}{n^{3}}%
		\rightarrow 0.
	\end{equation*}
	
	\textit{Step 4. The uniform convergence rate of }$\hat{\theta}(\tau )$%
	\textit{\ and }$\hat{m}(v,\tau )$\textit{\ for }$\tau \in \lbrack \tau
	_{L},\tau _{0}]$.
	
	For any $\tau \in \lbrack \tau _{L},\tau _{0}]$, there exists $l$ such that $%
	\tau \in \lbrack \tau _{l},\tau _{l-1}]$. Similar to Step 1, we can show $%
	\left\Vert \hat{\theta}(\tau )-\theta (\tau )\right\Vert \lesssim \tilde{\pi}%
	_{n}$ with the probability no less than $1-\frac{1}{n^{3}}$. With slight
	abuse of notations, define two classes of functions: 
	\begin{eqnarray*}
		\mathcal{F}_{r} &=&\{f_{r}:(T,W,Y)\longmapsto \frac{1}{\left\Vert \theta
			-\theta (\tau )\right\Vert }T1\{W^{\prime }\tilde{\theta}>\tilde{\eta}%
		_{n}\}[\rho _{\tau }(Y-W^{\prime }\theta )-\rho _{\tau }(Y-W^{\prime }\theta
		(\tau )) -\\
		&&\left. \varphi _{\tau }(Y-W^{\prime }\theta (\tau ))W^{\prime }\mathbf{(}%
		\theta -\theta (\tau )\mathbf{)}],\left\Vert \theta -\theta (\tau
		_{l})\right\Vert \lesssim \tilde{\pi}_{n},\left\Vert \tilde{\theta}-\theta
		(\tau _{l-1})\right\Vert \lesssim \pi _{n},\tau \in \lbrack \tau _{l},\tau
		_{l-1}]\right. \}
	\end{eqnarray*}%
	and 
	\begin{eqnarray*}
		\mathcal{F}_{h}&=&\{f_{h}:(T,W,Y)\longmapsto T1\{W^{\prime }\tilde{\theta}>%
		\tilde{\eta}_{n}\}\varphi _{\tau }(Y-W^{\prime }\theta (\tau ))W^{\prime }%
		\mathbf{\iota },\left\Vert \tilde{\theta}-\theta (\tau _{l-1})\right\Vert
		\lesssim \pi _{n},\\
		&&\tau \in \lbrack \tau _{l},\tau _{l-1}]\}.
	\end{eqnarray*}%
	Similar to $H_{n}(\theta (\tau _{l}),\tilde{\theta},\tau _{l},\tilde{\eta}%
	_{n})$ and $R_{n}(\theta ,\theta (\tau _{l}),\tilde{\theta},\tau _{l},\tilde{%
		\eta}_{n})$, we shall define 
	\begin{equation*}
		H_{n}(\theta (\tau ),\tilde{\theta},\tau ,\tilde{\eta}_{n})=\frac{1}{n}%
		\sum_{i=1}^{n}T_{i}1\{\hat{W}_{i}^{\prime }\tilde{\theta}>\tilde{\eta}%
		_{n}\}\varphi _{\tau }(Y_{i}-\hat{W}_{i}^{\prime }\theta (\tau ))\hat{W}_{i}
	\end{equation*}%
	and%
	\begin{eqnarray*}
		R_{n}(\theta ,\theta (\tau ),\tilde{\theta},\tau ,\tilde{\eta}_{n})
		&=&\frac{1}{\left\Vert \theta -\theta (\tau )\right\Vert }[S_{n}(\theta ,\tilde{\theta}%
		,\tau ,\tilde{\eta}_{n})-S_{n}(\theta (\tau ),\tilde{\theta},\tau ,\tilde{%
			\eta}_{n})-H_{n}(\theta (\tau ),\tilde{\theta},\tau ,\tilde{\eta}%
		_{n})^{\prime }\\
		&&\times(\theta -\theta (\tau ))].
	\end{eqnarray*}%
	As $E[f_{r}^{2}(T,\hat{W},Y)]\lesssim J$ and $E[f_{h}^{2}(T,\hat{W}%
	,Y)]\lesssim J$, applying Lemma S1 with $x=3\ln n$ yields%
	\begin{equation*}
		\sup_{\tau \in \lbrack \tau _{l},\tau _{l-1}]}\sup_{\left\Vert \tilde{\theta}%
			-\theta (\tau _{l-1})\right\Vert \lesssim \pi _{n}}\sup_{\left\Vert \theta
			-\theta (\tau )\right\Vert \lesssim \tilde{\pi}_{n}}\left\vert R_{n}(\theta
		,\theta (\tau ),\tilde{\theta},\tau ,\tilde{\eta}_{n})-E(R_{n}(\theta
		,\theta (\tau ),\tilde{\theta},\tau ,\tilde{\eta}_{n}))\right\vert
	\end{equation*}
	\begin{equation*}
		=\sup_{f_{r}\in \mathcal{F}_{r}}\left\vert \frac{1}{n}\sum_{i=1}^{n}\left[
		f_{r}(T_{i},\hat{W}_{i},Y_{i})-E(f_{r}(T_{i},\hat{W}_{i},Y_{i}))\right]
		\right\vert \lesssim \pi _{n}
	\end{equation*}%
	with the probability no less than $1-\frac{1}{n^{3}}$, and 
	\begin{eqnarray*}
		&&\sup_{\tau \in \lbrack \tau _{l},\tau _{l-1}]}\sup_{\left\Vert \tilde{%
				\theta}-\theta (\tau _{l-1})\right\Vert \lesssim \pi _{n}}\left\Vert
		H_{n}(\theta (\tau ),\tilde{\theta},\tau ,\tilde{\eta}_{n})-E[H_{n}(\theta
		(\tau ),\tilde{\theta},\tau ,\tilde{\eta}_{n})]\right\Vert \\
		&=&\sup_{f_{h}\in \mathcal{F}_{h}}\left\vert \frac{1}{n}\sum_{i=1}^{n}\left[
		f_{h}(T_{i},\hat{W}_{i},Y_{i})-E(f_{h}(T_{i},\hat{W}_{i},Y_{i}))\right]
		\right\vert \lesssim \pi _{n}
	\end{eqnarray*}%
	with the probability no less than $1-\frac{1}{n^{3}}$.
	
	Based on the fact that%
	\begin{eqnarray*}
		0 &\geq &-\left[\left\vert R_{n}(\hat{\theta}(\tau ),\theta (\tau ),\hat{\theta}%
		(\tau _{l-1}),\tau ,\tilde{\eta}_{n})-E\left( R_{n}(\hat{\theta}(\tau
		),\theta (\tau ),\hat{\theta}(\tau _{l-1}),\tau ,\tilde{\eta}_{n})\right)
		\right\vert  \right.\\
		&&+\left.\left\Vert H_{n}(\theta (\tau ),\tilde{\theta},\tau ,\tilde{\eta}%
		_{n})-E\left( H_{n}(\theta (\tau ),\tilde{\theta},\tau ,\tilde{\eta}%
		_{n})\right) \right\Vert +c_{\underline{s}}J^{1/2}\left( J^{-\kappa
		}+n^{-1/2}\right)\right] \\
		&&\times \left\Vert \hat{\theta}(\tau )-\theta (\tau )\right\Vert +\frac{1}{2%
		}\underline{c}_{\Omega }\left\Vert \hat{\theta}(\tau )-\theta (\tau
		)\right\Vert ^{2},
	\end{eqnarray*}%
	we obtain 
	\begin{equation*}
		P\left( \left. \sup_{\tau \in \lbrack \tau _{l},\tau _{l-1}]}\left\Vert \hat{%
			\theta}(\tau )-\theta (\tau )\right\Vert >c_{\theta }\pi _{n}\right\vert
		\Omega _{l-1}\right) \leq \frac{2}{n^{3}}
	\end{equation*}%
	for $l=1,\cdots ,L$. As a consequence, 
	\begin{eqnarray*}
		P\left( \sup_{\tau \in \lbrack \tau _{L},\tau _{0}]}\left\Vert \hat{\theta}%
		(\tau )-\theta (\tau )\right\Vert >c_{\theta }\pi _{n}\right) 
		&\leq&\sum_{l=1}^{L}P\left( \left. \sup_{\tau \in \lbrack \tau _{l},\tau
			_{l-1}]}\left\Vert \hat{\theta}(\tau )-\theta (\tau )\right\Vert >c_{\theta
		}\pi _{n}\right\vert \Omega _{l-1}\right) \\
		&\leq& \frac{2L}{n^{3}}=o(1),
	\end{eqnarray*}%
	and 
	\begin{eqnarray*}
		&&\sup_{\tau \in \lbrack \tau _{L},\tau _{0}]}\sup_{v\in \mathcal{\bar{V}}%
		}\left\vert \hat{m}(v,\tau )-m(v,\tau )\right\vert \\
		&\leq &\sup_{\tau \in \lbrack \tau _{L},\tau _{0}]}\sup_{v\in \mathcal{\bar{V%
		}}}\left\vert P_{J}(v)^{\prime }(\hat{\delta}(\tau )-\delta (\tau
		))\right\vert +\sup_{\tau \in \lbrack \tau _{L},\tau _{0}]}\sup_{v\in 
			\mathcal{\bar{V}}}\left\vert P_{J}(v)^{\prime }\delta (\tau )-m(v,\tau
		)\right\vert \\
		&=&O_{p}\left( \zeta _{0}\pi _{n}\right) +O_{p}(J^{-\kappa }).
	\end{eqnarray*}
	
	In sum,%
	\begin{equation*}
		\sup_{\tau \in \lbrack \tau _{L},\tau _{U}]}\left\Vert \hat{\theta}(\tau
		)-\theta (\tau )\right\Vert =O_{p}\left( \sqrt{\frac{J\ln n}{n}}\right)
	\end{equation*}
	
	\noindent and%
	\begin{equation*}
		\sup_{v\in \mathcal{\bar{V}}}\sup_{\tau \in \lbrack \tau _{L},\tau _{U}]}|%
		\hat{m}(v,\tau )-m(v,\tau )|=O_{p}\left( \zeta _{0}\sqrt{\frac{J\ln n}{n}}%
		\right) +O_{p}(J^{-\kappa })
	\end{equation*}%
\end{proof}
\begin{proof}[Proof of Theorem \ref{Thm4}]
	Define%
	\begin{eqnarray*}
		e_{0}(\tau ) &=&\frac{1}{n}\sum_{i=1}^{n}T_{i}\hat{D}_{\tau i}\varphi _{\tau
		}(Y_{i}-\hat{W}_{i}^{\prime }\hat{\theta}(\tau ))\hat{W}_{i} \\
		e_{1}(\tau ) &=&H_{n}(\hat{\theta}(\tau ),\hat{\theta}(\tau _{l-1}),\tau ,%
		\tilde{\eta}_{n})-E\left( H_{n}(\hat{\theta}(\tau ),\hat{\theta}(\tau
		_{l-1}),\tau ,\tilde{\eta}_{n})\right) -H_{n}(\theta (\tau ),\hat{\theta}%
		(\tau _{l-1}),\tau ,\tilde{\eta}_{n}) \\
		&&+E\left( H_{n}(\theta (\tau ),\hat{\theta}(\tau _{l-1}),\tau ,\tilde{\eta}%
		_{n})\right) ; \\
		e_{2}(\tau ) &=&\frac{1}{n}\sum_{i=1}^{n}T_{i}\hat{D}_{\tau i}\varphi _{\tau
		}(Y_{i}-\hat{W}_{i}^{\prime }\theta (\tau ))\hat{W}_{i}-\frac{1}{n}%
		\sum_{i=1}^{n}T_{i}D_{\tau i}(\tilde{\eta})\varphi _{\tau }(Y_{i}-X_{i}^{\prime
		}\beta (\tau )-R_{i}\alpha (\tau ) \\
		&&-m(V_{i},\tau ))W_{i}-\frac{1}{n}\sum_{i=1}^{n}E[TD_{\tau }(\tilde{\eta} )f_{Y^{\ast
			}|X,R,V}(Q_{Y^{\ast }|X,R,V}(\tau ))Wm^{\prime }(V,\tau )\\
		&&\phi(R,Z,R_{i},Z_{i})|R_{i},Z_{i}]; \\
		e_{3}(\tau ) &=&E\left[ \frac{1}{n}\sum_{i=1}^{n}T_{i}\hat{D}_{\tau
			i}(\varphi _{\tau }(Y_{i}-\hat{W}_{i}^{\prime }\hat{\theta}(\tau ))-\varphi
		_{\tau }(Y_{i}-\hat{W}_{i}^{\prime }\theta (\tau ))\hat{W}_{i})\right]
		-\Omega (\tau )(\hat{\theta}(\tau )-\theta (\tau )).
	\end{eqnarray*}%
	Notice that 
	\begin{eqnarray*}
		&&\frac{1}{n}\sum_{i=1}^{n}T_{i}\hat{D}_{\tau i}\varphi _{\tau }(Y_{i}-\hat{W%
		}_{i}^{\prime }\hat{\theta}(\tau ))\hat{W}_{i} \\
		&=&\Omega (\tau )(\hat{\theta}(\tau )-\theta (\tau ))+\frac{1}{n}%
		\sum_{i=1}^{n}T_{i}D_{\tau i}(\tilde{\eta} )\varphi _{\tau }(Y_{i}-X_{i}^{\prime
		}\beta (\tau )-R_{i}\alpha (\tau )-m(V_{i},\tau ))W_{i} \\
		&&+\frac{1}{n}\sum_{i=1}^{n}E[TD_{\tau }(\tilde{\eta} )f_{Y^{\ast
			}|X,R,V}(Q_{Y^{\ast }|X,R,V}(\tau ))Wm^{\prime }(V,\tau )\phi
		(R,Z,R_{i},Z_{i})|R_{i},Z_{i}]\\
		&&+e_{1}(\tau )+e_{2}(\tau )+e_{3}(\tau ).
	\end{eqnarray*}%
	It follows that
	\begin{equation}
		\Omega (\tau )\sqrt{n}(\hat{\theta}(\tau )-\theta (\tau ))=-\frac{1}{\sqrt{n}%
		}\sum_{i=1}^{n}G_{1i}(\tau )-\frac{1}{\sqrt{n}}\sum_{i=1}^{n}G_{2i}(\tau )+%
		\sqrt{n}e_{n}(\tau )  \tag{A4}\label{A4}
	\end{equation}%
	uniformly in $\tau \in \lbrack \tau _{L},\tau _{U}]$, where $e_{n}(\tau
	)=e_{0}(\tau )-e_{1}(\tau )-e_{2}(\tau )-e_{3}(\tau )$.
	
	Following the proof of Theorem 3.3 in Koenker and Bassett (1978), we have 
	\begin{equation*}
		\sup_{\tau \in \lbrack \tau _{L},\tau _{U}]}\left\Vert \frac{1}{n}%
		\sum_{i=1}^{n}\hat{D}_{\tau i}\varphi _{\tau }(Y_{i}-\hat{W}_{i}^{\prime }%
		\hat{\theta}(\tau ))\hat{W}_{i}\right\Vert =O_{p}\left( \frac{\xi _{0}J}{n}%
		\right) .
	\end{equation*}%
	Under Assumption A7, $\sup_{\tau \in \lbrack \tau _{L},\tau _{U}]}\left\Vert
	e_{0}(\tau )\right\Vert =o_{p}(n^{-1/2})$.
	Next, consider a class of functions, 
	\begin{eqnarray*}
		\mathcal{F}_{\bar{h}} &=&\{f_{\bar{h}}:(T,W,Y)\longmapsto T1\{W^{\prime }%
		\tilde{\theta}>\tilde{\eta}_{n}\}(1\{Y-W^{\prime }\theta <0\}-1\{Y-W^{\prime
		}\theta (\tau )<0\})W^{\prime }\mathbf{\iota }, \\
		&&\left\Vert \theta -\theta (\tau )\right\Vert \lesssim \pi _{n},\left\Vert 
		\tilde{\theta}-\theta (\tau _{l-1})\right\Vert \lesssim \pi _{n},\tau \in
		\lbrack \tau _{L},\tau _{U}],\mathbf{\iota \in }\mathcal{S}^{\dim (\theta
			)-1}\}.
	\end{eqnarray*}%
	The envelope of $\mathcal{F}_{\bar{h}}$ is bounded by $\zeta _{0}$ and the
	uniform entropy number of $\mathcal{F}_{\bar{h}}$ is bounded by $J$. After a
	simple calculation, one can obtain $\sup_{f_{\bar{h}}\in \mathcal{F}_{\bar{h}%
	}}E\left( f_{\bar{h}}^{2}\right) \lesssim J\zeta _{0}\pi _{n}$. By Lemma 21
	of Belloni et al. (2019), 
	\begin{eqnarray*}
		&&\sup_{\tau \in \lbrack \tau _{L},\tau _{U}]}\sup_{\left\Vert \theta
			-\theta (\tau )\right\Vert \lesssim \pi _{n}}\left\Vert H_{n}(\theta ,\hat{%
			\theta}(\tau _{l-1}),\tau ,\tilde{\eta}_{n})-E(H_{n}(\theta ,\hat{\theta}%
		(\tau _{l-1}),\tau ,\tilde{\eta}_{n}))\right. \\
		&&\left.-H_{n}(\theta (\tau ),\hat{\theta}%
		(\tau _{l-1}),\tau ,\tilde{\eta}_{n}) +E(H_{n}(\theta (\tau ),\hat{\theta}(\tau _{l-1}),\tau ,\tilde{\eta}%
		_{n}))\right\Vert\\
		&\lesssim& \sup_{f_{\bar{h}}\in \mathcal{F}_{\bar{h}}}\left\vert \frac{1}{n}%
		\sum_{i=1}^{n}\left[ f_{\bar{h}}(T_{i},\hat{W}_{i},Y_{i})-E(f_{\bar{h}%
		}(T_{i},\hat{W}_{i},Y_{i}))\right] \right\vert \\
		&=&O_{p}\left( \frac{J^{3/4}\zeta _{0}^{1/2}(\ln n)^{3/4}}{n^{3/4}}\right) +O_{p}\left( \frac{%
			J\zeta _{0}\ln n}{n}\right) .
	\end{eqnarray*}%
	It follows that $\sup_{\tau \in \lbrack \tau _{L},\tau _{U}]}\left\Vert
	e_{1}(\tau )\right\Vert =O_{p}\left( \frac{J^{3/4}\zeta _{0}^{1/2}(\ln
		n)^{3/4}}{n^{3/4}}\right) +O_{p}\left( \frac{J\zeta _{0}\ln n}{n}\right) $.
	
	By Lemma S2 and Lemma S3,
    \begin{eqnarray*} 
    	\sup_{\tau \in \lbrack \tau_{L},\tau _{U}]}\left\Vert e_{2}(\tau )\right\Vert 
    	&=&O_{p}\left(J^{1/2-\kappa }\right) +O_{p}\left( \sqrt{\frac{J^{1-\kappa }\ln n}{n}}%
	    \right)+O_{p}\left( \sqrt{\frac{J\pi _{1n}\ln n}{n}}\right) \\
	    &&+O_{p}\left(\sqrt{\frac{(\tilde{\eta}_{n}-\tilde{\eta} )J\ln n}{n}}\right) +O_{p}\left( \frac{%
		J\zeta _{0}\ln n}{n}\right) +o_{p}(n^{-1/2}) 
   \end{eqnarray*}	
		and 
	\begin{eqnarray*}
		\sup_{\tau \in \lbrack\tau _{L},\tau _{U}]}\left\Vert e_{3}(\tau )\right\Vert 
		&=&O_{p}\left( \frac{J\zeta _{0}\ln n}{n}\right) +O_{p}\left( (\tilde{\eta}_{n}-\tilde{\eta} )\sqrt{\frac{%
			J\ln n}{n}}\right) +O_{p}\left( \pi _{1n}\sqrt{\frac{J\ln n}{n}}\right)\\
		&&+O_{p}\left( J^{-\kappa }\sqrt{\frac{J\ln n}{n}}\right). 
   \end{eqnarray*}
	Hence, under Assumption A7,
	\begin{eqnarray*}
		\sup_{\tau \in \lbrack \tau _{L},\tau _{U}]}\left\Vert e_{n}(\tau
		)\right\Vert &=&O_{p}\left( \frac{J^{3/4}\zeta _{0}^{1/2}(\ln n)^{3/4}}{%
			n^{3/4}}\right) +O_{p}\left( J^{1/2-\kappa }\right) +O_{p}\left( \sqrt{\frac{%
				J^{1-\kappa }\ln n}{n}}\right)
		\\
		&& +O_{p}\left( \frac{J\zeta _{0}\ln n}{n}\right)+O_{p}\left( \sqrt{\frac{(\tilde{\eta}_{n}-\tilde{\eta} )J\ln n}{n}}\right)
		+O_{p}\left( \sqrt{\frac{J\pi _{1n}\ln n}{n}}\right)\\
		&=&o_{p}(n^{-1/2}).
	\end{eqnarray*}
	Up to now, collecting (\ref{A4}) and all terms above gives 
	\begin{eqnarray*}
		\sqrt{n}(\hat{\theta}(\tau )-\theta (\tau )) &=&-\frac{1}{\sqrt{n}}%
		\sum_{i=1}^{n}\Omega ^{-1}(\tau )G_{1i}(\tau )-\frac{1}{\sqrt{n}}%
		\sum_{i=1}^{n}\Omega ^{-1}(\tau )G_{2i}(\tau )+\sqrt{n}\Omega ^{-1}(\tau
		)e_{n}(\tau ) \\
		&=&\frac{1}{\sqrt{n}}\sum_{i=1}^{n}\xi _{i}(\tau )+o_{p}(1).
	\end{eqnarray*}%
\end{proof}
\begin{proof}[Proof of Theorem \ref{Thm5}]
	Based on Theorem \ref{Thm4}, 
	\begin{equation*}
		\sqrt{n}\left( \dbinom{\hat{\beta}(\tau )}{\hat{\alpha}(\tau )}-\dbinom{%
			\beta (\tau )}{\alpha (\tau )}\right) =\frac{1}{\sqrt{n}}\sum_{i=1}^{n}S_{1}%
		\xi _{i}(\tau )+o_{p}(1).
	\end{equation*}%
	By the functional central limit theorem, $\sqrt{n}\left( \dbinom{\hat{\beta}%
		(\tau )}{\hat{\alpha}(\tau )}-\dbinom{\beta (\tau )}{\alpha (\tau )}\right) $
	converges to a mean zero Gaussian process with the covariance function\ as $%
	\Sigma _{\beta \alpha }(\tau _{1},\tau _{2})=E\left[ S_{1}\xi _{i}(\tau
	_{1})\xi _{i}(\tau _{2})^{\prime }S_{1}^{\prime }\right] $.
\end{proof}
\begin{proof}[Proof of Theorem \ref{Thm6}]
	Let $\hat{\theta}^{\ast }(\tau )=(\hat{\beta}^{\ast }(\tau )^{\prime },\hat{%
		\alpha}^{\ast }(\tau ),\hat{\delta}^{\ast }(\tau )^{\prime })^{\prime }$.
	Following the proof of Theorem \ref{Thm4}, we shall establish 
	\begin{equation*}
		\sqrt{n}\left( \hat{\theta}^{\ast }(\tau )-\theta (\tau )\right) =\frac{1}{%
			\sqrt{n}}\sum_{i=1}^{n}\varsigma _{i}\xi _{i}(\tau )+o_{p}(1).
	\end{equation*}%
	It follows that $\sqrt{n}\left( \dbinom{\hat{\beta}^{\ast }(\tau )}{\hat{%
			\alpha}^{\ast }(\tau )}-\dbinom{\beta (\tau )}{\alpha (\tau )}\right) =\frac{%
		1}{\sqrt{n}}\sum_{i=1}^{n}\varsigma _{i}S_{1}\xi _{i}(\tau )+o_{p}(1)$
	uniformly in $\tau \in \lbrack \tau _{L},\tau _{U}]$. Combining it with $%
	\sqrt{n}\left( \dbinom{\hat{\beta}(\tau )}{\hat{\alpha}(\tau )}-\dbinom{%
		\beta (\tau )}{\alpha (\tau )}\right) =\frac{1}{\sqrt{n}}\sum_{i=1}^{n}\xi
	_{i}(\tau )+o_{p}(1)$, we obtain%
	\begin{equation*}
		\sqrt{n}\left( \dbinom{\hat{\beta}^{\ast }(\tau )}{\hat{\alpha}^{\ast }(\tau
			)}-\dbinom{\hat{\beta}(\tau )}{\hat{\alpha}(\tau )}\right) =\frac{1}{\sqrt{n}%
		}\sum_{i=1}^{n}(\varsigma _{i}-1)S_{1}\xi _{i}(\tau )+o_{p}(1)
	\end{equation*}%
	uniformly in $\tau \in \lbrack \tau _{L},\tau _{U}]$. By the conditional
	central limit theorem in Theorem 2.96 of van der Vaart and Wellner (1996), $%
	\sqrt{n}\left( \dbinom{\hat{\beta}^{\ast }(\tau )}{\hat{\alpha}^{\ast }(\tau
		)}-\dbinom{\hat{\beta}(\tau )}{\hat{\alpha}(\tau )}\right) $ converges to a
	mean zero Gaussian process with the covariance function as $\Sigma _{\beta
		\alpha }(\tau _{1},\tau _{2})=E\left[ S_{1}\xi _{i}(\tau _{1})\xi _{i}(\tau
	_{2})^{\prime }S_{1}^{\prime }\right] $.
\end{proof}

\end{document}